\documentclass{lmcs}
\pdfoutput=1

\usepackage{lastpage}
\lmcsdoi{15}{4}{5}
\lmcsheading{}{\pageref{LastPage}}{}{}%
{Mar.~20,~2018}{Oct.~29,~2019}{}

\usepackage{amssymb}
\usepackage{amsfonts}
\usepackage{amsmath}
\usepackage{amsthm}
\usepackage{mathrsfs}
\usepackage{amscd}
\usepackage{bussproofs}
\usepackage[all]{xy}
\usepackage{comment}
\usepackage{cite}
\usepackage[textsize=tiny]{todonotes}
\usepackage{multicol}

\newtheorem{stw}[thm]{Proposition}
\theoremstyle{definition}
\newtheorem{defin}[thm]{Definition}

\newcommand{\subf}{\trianglelefteq}

\newcommand{\diam}{\diamondsuit}
\newcommand{\pair}[2]{\langle #1, #2 \rangle}
\newcommand{\set}[2]{\{#1 \ \ | \ \ #2 \}}
\newcommand{\supp}{\textrm{supp}}
\newcommand{\trip}[3]{\langle #1, #2, #3 \rangle}
\newcommand{\quadr}[4]{\langle #1, #2, #3, #4 \rangle}

\newcommand{\form}{{\mathcal{L}}_\mu}

\newcommand{\was}[1]{\{#1\}}
\newcommand{\nat}{\mathbb{N}}

\newcommand{\pargame}{\quadr{V}{V_{\exists}}{R}{\Omega}}

\newcommand{\tpaths}[1]{t(\phi)^*}
\newcommand{\conc}{^\frown}

\newcommand{\evalg}[1]{\mathcal{G}_{#1,\mathcal{K}}}

\newcommand{\PM}{\textrm{PM}}

\newcommand{\Win}{\textnormal{Win}}
\newcommand{\CTL}{\textnormal{CTL}}
\newcommand{\LTL}{\textnormal{LTL}}
\newcommand{\Atoms}{\mathbb{A}}
\newcommand{\aut}{{\textrm{Aut}}}
\newcommand{\freshpath}{\textsc{\#Path}}
\newcommand{\infsucc}{\textsc{InfSucc}}
\newcommand{\chain}{\textsc{Chain}}
\newcommand{\pred}{{\textrm{pred}}}
\newcommand{\aform}{\mathcal{L}_{\mu}^{\Atoms}}
\newcommand{\eqaform}{\mathcal{L}_{\mu}^{=}}
\newcommand{\ordaform}{\mathcal{L}_{\mu}^{<}}
\newcommand{\K}{{\mathcal{K}}}
\newcommand{\step}{\longrightarrow}
\newcommand{\sembr}[1]{[\![#1]\!]}
\newcommand{\pop}{{\textrm{pop}}}
\newcommand{\vform}{{\vec{\mathcal{L}}}_\mu}
\newcommand{\vaform}{{\vec{\mathcal{L}}}_{\mu}^{\Atoms}}
\newcommand{\veqaform}{{\vec{\mathcal{L}}}_{\mu}^{=}}
\newcommand{\vordaform}{{\vec{\mathcal{L}}}_{\mu}^{<}}
\newcommand{\evensucc}{\textsc{EvenSucc}}
\newcommand{\power}{\mathcal{P}}

\newcommand{\gsp}{\blacktriangleright}

\begin{document}

\title[Scalar and Vectorial $\mu$-calculus with Atoms]{Scalar and Vectorial $\mu$-calculus with Atoms${}^{\dagger}$}
\thanks{${}^{\dagger}$\ This is a revised and extended version of the conference paper~\cite{csl17}.}

\author[Bartek Klin]{Bartek Klin${}^*$}
\address{University of Warsaw, Dept.~of Mathematics, Informatics and Mechanics}	
\email{klin@mimuw.edu.pl}  
\thanks{${}^*$ Supported by the Polish NCN grant 2012/07/E/ST6/03026.}

\author[Mateusz \L{}e\l{}yk]{Mateusz \L{}e\l{}yk${}^{**}$}
\address{University of Warsaw, Dept. of Philosophy and Sociology}	
\email{mlelyk@uw.edu.pl}  
\thanks{${}^{**}$ Partially supported by the Polish NCN grant 2016/21/B/ST6/01505.}

\subjclass{F.4.1 Mathematical Logic--Temporal logic, D.2.4 Software/Program Verification--Model checking}
\keywords{modal $\mu$-calculus, sets with atoms}
	
\begin{abstract}
We study an extension of modal $\mu$-calculus to sets with atoms and we study its basic properties. Model checking is decidable on orbit-finite structures, and a correspondence to parity games holds. On the other hand, satisfiability becomes undecidable. We also show expressive limitations of atom-enriched $\mu$-calculi, and explain how their expressive power depends on the structure of atoms used, and on the choice between basic or vectorial syntax. 
\end{abstract}	
	
\maketitle

\section{Introduction}\label{sect_intro}

Modal $\mu$-calculus~\cite{kozen-mu,niw-mu,BS01,BS07} is perhaps the best known formalism for describing properties of labeled transition systems or Kripke models.
It combines a simple syntax with a mathematically elegant semantics and it is expressive enough to specify many interesting properties of systems. For example, the property ``[in the current state] the predicate $p$ holds, and there exists a transition path where it holds again some time in the future'' is defined by a $\mu$-calculus formula:
\[
	p \land \diam\mu X.(p\lor\diam X).
\]
Other similar formalisms, such as the logic CTL$^*$~\cite{CTLstar}, can be encoded in the modal $\mu$-calculus.
On the other hand, decision problems such as model checking (``does a given $\mu$-calculus formula hold in a given (finite) model?'') or satisfiability (``does a given formula hold in some model?'') are decidable.

Formulas of the $\mu$-calculus are built over some fixed vocabulary of basic predicates, such as $p$ above, whose semantics is provided in every model. In principle this vocabulary may be infinite, but this generality is hardly useful: since a formula of the $\mu$-calculus is a finite object, it may only refer to finitely many basic predicates.

In modelling systems, this finiteness may sometimes seem restrictive. Real systems routinely operate on data coming from potentially infinite domains, such as numbers or character strings. Basic predicates observed about a system may reasonably include ones like ``a number $n$ was input'', denoted here $p_n$, for every number $n$. If one considers properties such as ``there exists a transition path where the currently input number is input again some time in the future'', one is in trouble writing finite formulas to define them. Were infinitary connectives allowed in the formalism, the formula
\[
	\bigvee_{n\in\mathbb{N}}\left(p_n \land \diam\mu X.(p_n\lor\diam X)\right)
\]
would do the job; however, few good properties of the $\mu$-calculus transport to a naively construed infinitary setting. In particular, an obvious obstruction to any kind of decidability results is that a general infinitary formula is not readily presented as input to an algorithm.

In this paper we introduce $\mu$-calculus with atoms, where infinitary propositional connectives are allowed in a restricted form that includes formulas such as the one above. Roughly speaking, basic properties in formulas, and indeed whole Kripke models, are assumed to be built of {\em atoms} that come from a fixed infinite structure. Out of many possible structures of atoms, we concentrate on two: {\em equality atoms}, which come from a pure infinite set without any relations, and {\em ordered atoms}, where a total ordering on the atoms can be used to build Kripke models and formulas.

Boolean connectives in $\mu$-formulas are indexed by sets that are potentially infinite but are finitely definable in terms of atoms. This makes formulas finitely presentable. 

Several basic properties of the standard $\mu$-calculus hold in the atom-based setting. Syntax and semantics of the calculus is defined in a standard way. The model checking problem remains decidable, although this is not trivial, as formulas and models are now, strictly speaking, infinite. A correspondence between $\mu$-formulas and parity games with atoms holds. 

Some other properties of the classical calculus fail. The satisfiability problem becomes undecidable, and the atom-based generalization of the finite model property fails. A vectorial extension of the $\mu$-calculus, which in the classical setting is expressively equivalent to the basic ``scalar'' variant, here becomes strictly more expressive. 

Atomic $\mu$-calculi also turn out to be less expressive than might be expected. In particular, the property \freshpath{}, which says ``there exists a path where no basic predicate holds more than once'', although decidable, is not definable. This means that, unlike in the classical setting, an atomic extension of CTL${}^*$ (where explicit quantification over paths makes such properties easy to define) is not a fragment of the atomic $\mu$-calculus. As it turns out, for atomic CTL${}^*$ even the model-checking problem is undecidable.

Our approach is a part of a wider programme of extending various computational models to sets with atoms, which are a generalization of nominal sets~\cite{pitts-book}. For example, in~\cite{lmcs14} the classical notion of finite automaton was reinterpreted in the universe of sets with atoms, with a result related to register automata~\cite{FK94,NSV01}, an established model of automata over infinite alphabets. 

Temporal logics over structures extended with data from infinite domains, and their connections to various types of automata, have been extensively studied in the literature. For example, the linear time logic LTL has been extended with a {\em freeze quantifier}~\cite{DL09,DLN07,FS09,Seg06} which, for structures where every position is associated with a single data item, can store the current item for future reference. This can serve as a mechanism for detecting repeated data values (see also~\cite{DDG12,DFP16}). Another known idea is to extend temporal logics with local constraints over data from a fixed infinite domain, see e.g.~\cite{CL15,DDS07}. In~\cite{Fig12,JL11}, alternating register automata on data words and trees were studied, closely connected to $\mu$-calculi.

In all these works, the main goal is to study the decidability border for satisfiability (or nonemptiness, in the case of automata). To that end, the authors impose various, sometimes complex restrictions on their logics or automata, and inevitably limit their expressiveness to some extent. In contrast to that, our aim is to lift the classical modal $\mu$-calculus to data-equipped structures in the most syntactically economic way, and to achieve a relatively expressive formalism. As a price for this, the satisfiability problem quickly becomes undecidable. However, we believe that this does not disqualify the atomic $\mu$-calculus as a practical formalism: most applications of temporal logics in system verification rely only on solving the model checking problem, and that remains decidable here on a wide class of structures.

It would be easy to give up even the decidability of model checking. This was done in~\cite{PBEGW15}, in a setting very similar to ours, by extending a basic multimodal logic with infinitary boolean connectives subject only to a finite support condition. The resulting logic has few good properties except its huge expressive power (indeed, it can immediately encode our atomic $\mu$-calculus, and much more): its set of formulas is not even countable. Instead, our boolean connectives are subject to a more restrictive condition of orbit finiteness, an idea first used in~\cite{BP12} in the context of first-order logic with atoms.

Another branch of related work is centered around algebras of name-passing processes such as the $\pi$-calculus, and variants of the $\mu$-calculus aimed at specific transition systems induced by those algebras. This line of work was started in~\cite{Dam96}, where a version of the $\mu$-calculus for model checking properties of $\pi$-calculus processes was proposed, with a sound-and-complete proof system. Other efforts in this direction, resulting in logics fine-tuned to specific infinite models, include~\cite{FGMP03,Lin05}, and a more abstract calculus was proposed in~\cite{DNL08}. 

The closest related work in the literature is the {\em first-order $\mu$-calculus} of Groote et al.~\cite{GM99,GW05}, where the classical $\mu$-calculus is extended with quantification over data values from arbitrary infinite data domains, and with fixpoints parametrized by data. Our infinitary boolean connectives can be seen as a syntactic variant of quantification over data values, and indeed our atomic $\mu$-calculi in their more expressive, vectorial form can be seen as variants of the first-order $\mu$-calculus specialized to very particular, well-behaved infinite data domains. The main difference is in focus: while~\cite{GM99,GW05} aim for an expressive formalism over rich data domains, with little hope for general decidability results and with a focus on pragmatic usability, we insist on keeping model-checking decidable and delineate expressiveness limitations that arise as a result.

This work is a revised and extended version of the conference paper~\cite{csl17}. There, the basic (``scalar'') atomic $\mu$-calculus over equality atoms $\eqaform$ was introduced. Here we extend that definition to ordered atoms, and consider a vectorial extension of atomic $\mu$-calculus over both structures of atoms, resulting in four calculi altogether. As it turns out, each of the four has a different expressive power, but (in addition to the obviously similar definitions of syntax and semantics) they share several basic properties, including the decidability of model checking.

The structure of this paper is as follows. In Sect.~\ref{sect_modal_mu_and_related} we briefly recall the modal $\mu$-calculus and related logics in the classical setting. In Sect.~\ref{sect_atoms_basic} we recall the basics of sets with atoms, including the notions of finite support and orbit finiteness, focusing on two important cases: equality atoms and ordered atoms. In Sect.~\ref{sect_mu_atoms} we introduce the syntax and semantics of the atomic $\mu$-calculus both in a basic ``scalar'' and in a vectorial version, for both these structures of atoms. In Sects.~\ref{sect_model_checking}--\ref{sect_bisimulations} we show various basic properties of atomic $\mu$-calculi, including the decidability of model checking, and a correspondence to parity games and atomic bisimulation games. In Sect.~\ref{sect_undecidability} we focus on undecidability results. In Sect.~\ref{sect_separation} we define properties of Kripke models that separate the four calculi we study in this paper, and in Sect.~\ref{sect_limitations} we formulate a decidable and natural property that is undefinable in all of them. A brief list of future work directions is in Sect.~\ref{sect_conclusions}.

\noindent
{\bf Acknowledgments.} We are very grateful to A.~Facchini for collaboration in initial stages of this work,  to M.~Boja\'nczyk, W.~Czerwi\'nski, P.~Hofman, S.~Lasota, S.~Toru\'nczyk and B.~Wcis\l{}o for valuable discussions, and to anonymous reviewers for their thorough work and numerous insightful comments.


\section{$\mu$-calculus and related logics}\label{sect_modal_mu_and_related}
To fix the notation and terminology, we begin by recalling basic definitions and properties of the $\mu$-calculus in the classical setting. For a more detailed exposition see e.g.~\cite{niw-mu,BS01,BS07,yde,Thomas2004-THOALA-3,BraWal15}.

\begin{defin}[Syntax]\label{def_syntax}
		Let $\mathbb{P}$ be an infinite set of {\em basic predicates} and $\mathbb{X}$ an infinite set of \emph{variables}. The set $\form$ of $\mu$-calculus formulas is generated by the grammar:
		\[\phi::= p \mid X \mid \phi\vee\phi \mid \neg \phi\mid \diam\phi\mid \mu X.\phi \]
		where $p$ ranges over $\mathbb{P}$ and $X$ over $\mathbb{X}$. We only allow formulas $\mu X.\phi$ where $X$ occurs only positively in $\phi$ (i.e. under an even number of negations).
	\end{defin}

 We use the following standard abbreviations: 
 \begin{multicols}{2}
\begin{itemize}
\item $ \top := p\lor\neg p$,
\item $\phi\wedge \psi := \neg(\neg\phi\vee\neg\psi)$,
\item $\Box\phi := \neg\diam\neg\phi$,
\item $\nu X. \phi := \neg\mu X.\neg\phi[X:=\neg X]$.
\end{itemize}
\end{multicols}

 \begin{defin}[Kripke model]\label{def_krip_mod}
A \textit{Kripke model} is a triple 
\[
	\K = \trip{K}{\step^{\K}}{\models^{\K}}
\]
 such that
\begin{itemize}
\item $K$ is a set of {\em states},
\item ${\step^{\K}}\subseteq K\times K$ is a {\em transition relation},
\item ${\models^{\K}}\subseteq K\times \mathbb{P}$ is the {\em satisfaction relation} for basic predicates.
\end{itemize}
\end{defin}

Superscripts in $\step^{\K}$ and $\models^{\K}$ will be often omitted when no risk of confusion arises. As a convention, for a model denoted by a calligraphic letter such as $\K$, its set of states will be denoted by the corresponding italic letter.

Sometimes it will be notationally convenient to speak of the set of basic predicates that hold in a state of a model:
\[
	\pred^{\K}(x) = \{p\in\mathbb{P} \mid x\models^{\K}p\}.
\]
Again, the superscript will usually be omitted.

\begin{defin}[Semantics]\label{def_speln_stand}
Formulas of $\mu$-calculus are interpreted in the context of a Kripke model $\K$ and an environment, i.e., a partial function $\rho:\mathbb{X}\rightharpoonup {\mathcal{P}}(K)$. For any formula $\phi$, the interpretation $\sembr{\phi}_\rho\subseteq K$ is defined by induction:
\begin{itemize}
\item $\sembr{p}_\rho =  \{x\in K \mid x\models p\}$,
\item $\sembr{X}_\rho = \rho(X)$,
\item $\sembr{\neg{\phi}}_\rho = K\setminus \sembr{\phi}_\rho$,
\item $\sembr{\phi\vee\psi}_\rho = \sembr{\phi}_\rho\cup\sembr{\psi}_\rho$,
\item $\sembr{\diam \phi}_\rho =  \set{k\in K}{\exists s\in\sembr{\phi}_\rho.\ k\step s}$,
\item $\sembr{\mu X.\phi}_\rho = \mathrm{lfp}(F), \text{where } F(A)=\sembr{\phi}_{\rho[X\mapsto A]}$.
\end{itemize}
In the last clause, the least fixpoint is taken for a function $F:{\mathcal{P}}(K)\to{\mathcal{P}}(K)$. Subsets of $K$ form a complete lattice and (thanks to the positivity assumption in Definition~\ref{def_syntax}) the function $F$ is monotone, so by Tarski's theorem the least fixpoint exists and it arises as the union of an increasing chain of approximants:
\[
    A^0\subseteq A^1\subseteq A^2 \subseteq \cdots A^\omega\subseteq A^{\omega+1}\subseteq \cdots
\]
indexed by ordinal numbers, where
\begin{itemize}
\item $A^{\alpha+1}=F(\alpha)$ and
\item $A^\beta = \bigcup_{\alpha<\beta}A^\alpha$ for a limit ordinal $\beta$.
\end{itemize}
In particular, $A^0=\emptyset$.

We write $\sembr{\phi}$ (or $\sembr{\phi}^{\K}$ if $\K$ is less clear from context) instead of $\sembr{\phi}_\rho$ if $\rho$ is empty (i.e., nowhere defined).
We say that $\phi$ {\em holds} in a state $x\in K$ if $x\in\sembr{\phi}$ and denote it $x\models \phi$; this clearly agrees with the satisfaction relation $\models$ for basic predicates.
\end{defin}

\begin{exa}
The formula $\mu X. (p\lor \diam X)$ holds for a state $x$ in a Kripke model $\mathcal{K}$ if and only if there is a finite path from $x$ to a state where $p$ holds. The formula $\nu X.(\Box X \wedge \mu Y. (p \lor \Box Y))$ holds in those states $x$ where $p$ holds infinitely often on every path from $x$.
\end{exa}

For applications in system verification, the following two problems are considered:

\noindent
{\bf Model checking:} given a finite Kripke model $\K$, a state $x\in K$ and a formula $\phi\in\form$, decide if $x\models\phi$.

\noindent
{\bf Satisfiability:} given a formula $\phi\in\form$, decide if there exists a Kripke model $\K$ and a state $x\in K$ such that $x\models\phi$.

It is very easy to see that model checking is decidable: in a finite Kripke model, one can compute the semantics of a formula inductively, directly from the definition. A more efficient procedure can be derived from a correspondence of $\mu$-calculus with parity games.

Some well-known logics used in system verification can be translated into fragments of the $\mu$-calculus. We give two important examples:

\begin{defin}[$\CTL^*$]\label{def:ctl}
In the logic $\CTL^*$ we distinguish state formulas $\Phi$ and path formulas $\phi$, formed according to the following grammar:
\[
\Phi ::= p \mid \Phi\vee\Phi \mid \neg\Phi \mid \exists\phi \qquad\qquad
\phi ::= \Phi \mid \phi\vee\phi \mid \neg\phi \mid \phi {\sf U}\phi \mid {\sf X}\phi
\]
where $p$ comes from some fixed set of propositional variables.
\end{defin}
Standard notational conventions include: 
 \begin{multicols}{2}
\begin{itemize}
\item $\forall\phi = \neg\exists\neg\phi$, 
\item $\phi {\sf R}\psi = \neg(\neg\phi {\sf U}\neg\psi)$, 
\item ${\sf F}\phi = \top {\sf U} \phi$, 
\item ${\sf G}\phi = \neg {\sf F}\neg\phi$.
\end{itemize}
\end{multicols}

$\CTL^*$ formulas are interpreted in Kripke models $\K$: state formulas are interpreted over states, and path formulas over paths. In the following definition of a satisfaction relation $\models$, for a path $\pi=\langle x_0,x_1,x_2,\ldots\rangle$, by $\pi[n..]$ we denote the subpath starting at $x_n$.
\begin{itemize}
\item ${x}\models \exists\phi\iff$ for some path $\pi$ starting at $x$, $\pi\models\phi$.
\item $\pi\models \Phi\iff \pi[0]\models \Phi$.
\item $\pi\models {\sf X}\phi\iff \pi[1..]\models \phi$
\item $\pi\models \phi {\sf U} \psi \iff$ there exists $j\geq 0$ s.t.~$\pi[j..]\models \psi$ and for all $i<j$, $\pi[i..]\models \phi$.
\end{itemize}
Boolean connectives are interpreted as expected.

The logic $\LTL$ is the fragment of $\CTL^*$ where the symbol $\exists$ does not occur, with semantics inherited from $\CTL^*$.
Usually LTL (where the distinction between state and path formulas disappears) is interpreted in models that are infinite paths. Slightly more generally (and more conveniently for our purposes), we will interpret them over pointed {\em deterministic} Kripke models, i.e., ones where for each $x\in K$ there is exactly one $y\in K$ such that $x\step y$. This makes little difference, since in such a model every state uniquely determines an infinite path.

Sometimes it is convenient to consider a syntactic extension of the $\mu$-calculus to a {\em vectorial} form. There, the fixpoint construct $\mu X.\phi$ is generalized to one of the form
\begin{align}\label{eq:vectorial}
	\phi = \mu X_i.\left\{\begin{array}{rl}X_1. &\!\!\phi_1 \\ \vdots \\ X_n. &\!\!\phi_n\end{array}\right\} 
\end{align}
where $\{X_1,\ldots,X_n\}\subseteq\mathbb{X}$ (where we assume $X_i\neq X_j$ for $i\neq j$) is any finite set of variables, $X_i$ is some element of that set, and $\phi_1,\ldots,\phi_n$ are formulas where all of $X_1,\ldots,X_n$ are considered bound (and, as before, occur only positively). 

There are a few intuitive notions concerning the body of the formula in~\eqref{eq:vectorial}, and it will be useful to keep them in mind when we extend the vectorial $\mu$-calculus with atoms in Section~\ref{sect_mu_atoms}. An expression 
\[	
	X_j.\phi_j
\]
may be seen as an {\em equation} (where the $X_j$ is the left-hand side, and the $\phi_j$ the right-hand side), and the set
\[
\Phi = \left\{X_j.\phi_j\right\}_{j=1..n}
\]
is a {\em system of equations} where each variable appears on the left-hand side exactly once. Alternatively, it can be seen as a function
from the finite set $\{X_j\}_{j=1..n}$ to the set of formulas. The variable $X_i$ in $\phi$, which we call the {\em entry variable} of the fixpoint formula, must belong to that finite set. Then, striving for syntax that fits in a single line of text, the formula in~\eqref{eq:vectorial} may be rewritten as:
\begin{align}\label{eq:vectorial2}
	\phi = \mu X_i.\Phi = \mu X_i.\left\{X_j.\phi_j\right\}_{j=1..n}.
\end{align}
 
The set of vectorial $\mu$-calculus formulas will be denoted by $\vform$. When we want to distinguish the original calculus $\form$ from the vectorial one, we will call it the {\em scalar $\mu$-calculus}.

For the semantics of a formula $\phi$ as above, in the context of a model $\K$ and an environment $\rho:\mathbb{X}\rightharpoonup {\mathcal{P}}(K)$, consider a function
\[
	F:(\mathcal{P}(K))^n\to(\mathcal{P}(K))^n
\]
defined by:
\[
	\left(F(A_1,\ldots,A_n)\right)_j = \sembr{\phi_j}_{\rho[X_1\mapsto A_1,\ldots,X_n\mapsto A_n]} \quad \text{for } j=1,\ldots,n,
\]
then, for $(B_1,\ldots,B_n)$ the least fixed point of $F$, define
\[
	\sembr{\phi}_{\rho} = B_i.
\]
It is not difficult to see that for $n=1$ this specializes to the original semantics in Definition~\ref{def_speln_stand}.

In the classical setting, vectorial modal $\mu$-calculus is expressively equivalent to the scalar one, i.e., every formula can be rewritten to a semantically equivalent scalar form, using the so-called Beki\'c principle. The vectorial calculus is therefore mostly considered for convenience, as its formulas seem to be more succint, and they correspond more closely to alternating tree automata. As we shall see in Section~\ref{sect_separation}, in the atomic setting the difference is more significant, and adding vectorial formulas will extend the expressivity of atomic $\mu$-calculus.

\section{Sets with atoms}\label{sect_atoms_basic}
We now recall the basic notions and results concerning sets with atoms, beginning with the case of {\em equality atoms}, also known as nominal sets~\cite{pitts-book}. There are several essentially equivalent ways to introduce these; we follow the set-theoretic presentation of~\cite{gabbay-pitts}, culminating in the notion of orbit-finite sets~\cite{lmcs14,pitts-book} and computable operations on them. Then we recall a more general setting, advocated in~\cite{lmcs14}, where atoms are equipped with a nontrivial relational structure, in particular a total order.

\begin{rem}
When developing a theory of "objects with atoms" (like automata with atoms or modal $\mu$-calculus with atoms) there seems to be two natural ways to proceed: one can either (i) stick to the universe of (sufficiently well-behaved) sets with atoms and work exclusively within this universe (which corresponds to working in a formal system such as ZFA), or (ii) work in a broader universe of all sets and within this universe define what does it mean to be a (sufficiently well-behaved) set with atom. We regard the second strategy as more convenient to our immediate goals so will consistently follow it.
\end{rem}

\subsection{Equality atoms}

Fix a countably infinite set $\Atoms$, whose elements we shall call {\em atoms}. A bijection on $\Atoms$ will be called an {\em atom automorphism}, and the group of atom automorphisms is denoted $\aut(\Atoms)$. 

Loosely speaking, a set with atoms is a set that can have atoms, or other sets with atoms, as elements. Formally, the universe $\mathcal{U}^{\Atoms}$ of sets with atoms is defined by a von Neumann-like hierachy, by transfinite induction on ordinal numbers $\alpha$:
\[
\mathcal{U}^{\Atoms}_0 = \emptyset, \qquad\qquad
\mathcal{U}^{\Atoms}_{\alpha+1} = \mathcal{P}(\mathcal{U}^{\Atoms}_{\alpha}) + \Atoms, \qquad\qquad
\mathcal{U}^{\Atoms}_{\beta} = \bigcup_{\alpha<\beta}\mathcal{U}_{\alpha}^{\Atoms}\quad\text{for $\beta$ a limit ordinal},
\]
where $+$ denotes disjoint union of sets.

We are interested in sets that only depend on a finite number of atoms, in the following sense.
Atom automorphisms 
act on the universe $\mathcal{U}^{\Atoms}$ by consistently renaming all atoms in a given set. Formally, this is again defined by transfinite induction. This defines a group action: 
\[
	\_\cdot\_: {\mathcal{U}^{\Atoms}}\times\aut(\Atoms) \to {\mathcal{U}^{\Atoms}}.
\]
For a finite set $S\subset\Atoms$, let $\aut_S(\Atoms)$ be the group of those automorphisms of $\Atoms$ that fix every element of $S$.
We say that $S$ {\em supports} a set $x$ if $x\cdot\pi=x$ for every $\pi\in\aut_S(\Atoms)$. A set is {\em equivariant} if it is supported by the empty set. If $x$ has a finite support then it has the least finite support (see~\cite[Cor.~9.4]{lmcs14} for a proof), denoted $\supp(x)$.

\begin{rem}
In~\cite{gabbay-pitts,pitts-book} a slightly different variant of equality atoms is developed, where $\aut(\Atoms)$ is taken to be the group of {\em finite} bijections, i.e. those that fix almost all atoms. As argued in~\cite[Sec.~6.2]{pitts-book}, this difference is irrelevant as far as finitely supported sets are concerned: exactly the same finite supports exist under both definitions. In the context of finite bijections as atom automorphisms, existence of least supports was first proved in~\cite[Prop.~3.4]{gabbay-pitts}. We allow all atom bijections as automorphisms since this approach generalizes more naturally to richer atom structures: anticipating ordered atoms (see~Section~\ref{sec:ordered-atoms} below), note that there are no nontrivial finite monotone bijections on the total order of rational numbers.
\end{rem}

From now on, we shall only consider sets with atoms that are {\em hereditarily finitely supported}, i.e., ones that have a finite support, whose every element has some finite support and so on.

Relations, functions etc.~are sets in the standard sense, so the notions of support and equivariance applies to them as well. Unfolding the definitions, for equivariant sets $X$ and $Y$, a relation $R\subseteq X\times Y$ is equivariant if 
\[
\pair{x}{y}\in R \text{ implies } \pair{x\cdot\pi}{y\cdot\pi}\in R
\]
and a function $f:X\to Y$ is equivariant if 
\[
   f(x\cdot\pi) = f(x)\cdot\pi
\]
for every $\pi\in\aut(\Atoms)$.

For any $x$ with atoms, the {\em $S$-orbit} of $x$ is the set
$\{x\cdot\pi \mid \pi\in\aut_S(\Atoms)\}$.
Note that if $S$ supports $x$ then the $S$-orbit of $x$ is the singleton $\{x\}$. We shall write $x\sim_S y$ to say that $x$ and $y$ belong to the same $S$-orbit, and we shall omit the subscript and write simply $x\sim y$ for empty $S$. Similarly, an $\emptyset$-orbit will simply be called an orbit.

For any $S$, $S$-orbits form a partition of the universe $\mathcal{U}$. Moreover, for any $S$-supported set $X$, the $S$-orbits of its elements form a partition of $X$. We call such $X$ {\em $S$-orbit-finite} if it is a union of finitely many $S$-orbits. If $S\subseteq T$ are finite, $S$ supports $X$ and $X$ is $S$-orbit-finite, then ($T$ supports $X$ and) $X$ is also $T$-orbit-finite. Thanks to this observation, we may drop the qualifier and simply call $X$ {\em orbit-finite}, meaning ``$S$-orbit-finite for any/every $S$ that supports $X$''. It is not difficult to check that every finitely supported subset of an orbit-finite set is orbit-finite.

\begin{exa}\label{exa:sets-with-atoms}
\begin{itemize}
\item Any classical set (without atoms) is an equivariant set with atoms. Since its every element is also equivariant, it forms its own orbit. Therefore, a classical set is orbit-finite if and only if it is finite.
\item An atom $a\in\Atoms$ has no elements, and it is supported by $\{a\}$. Every finite set of atoms $S\subseteq\Atoms$ is supported by $S$, and every element of it forms a singleton $S$-orbit. Its complement $\Atoms\setminus S$ is also supported by $S$, and it is a single $S$-orbit. A subset of $\Atoms$ that is neither finite nor co-finite is not finitely supported, so we do not consider it a legal set with atoms.
\item The set $\Atoms$ of
atoms, the set ${\Atoms}\choose{2}$ 
of two-element sets of atoms, and the set $\Atoms^2$ of all ordered pairs of atoms, are equivariant sets. The first two have a single orbit each, and the last one is a union of two orbits:
\[
	\Atoms^2 = \{\pair{a}{a}\mid a\in\Atoms\} \cup \{\pair{a}{b}\mid a\neq b\in\Atoms\}.
\]
Similarly, $\Atoms^n$ is orbit-finite for every $n\in\mathbb{N}$. The set $\Atoms^*$ of finite sequences of atoms is hereditarily finitely supported, but not orbit-finite. 
\item The powerset $\mathcal{P}(\Atoms)$ is equivariant itself, but it contains elements that are not finitely supported, and therefore it is not considered a legal set with atoms. However, for every legal set with atoms $X$, the set $\mathcal{P}_{\text{fs}}(X)$ of {\em finitely supported} subsets of $X$ has the same support as $X$ and is a legal set of atoms.
\item There are four equivariant binary relations on $\Atoms$: the empty relation, the equality relation, the inequality relation and the full relation. 
\item 
There is no equivariant function from ${\Atoms}\choose{2}$ to $\Atoms$, but 
$
	\{\pair{\{a,b\}}{a} \mid a\neq b\in\Atoms\}
$
is a legal equivariant relation, and the function constant at an atom $a$ is supported by $\{a\}$. The only equivariant function from $\Atoms$ to $\Atoms$ is identity. The only equivariant functions from $\Atoms^2$ to $\Atoms$ are the two projections, and the only equivariant function from $\Atoms$ to $\Atoms^2$ is the diagonal $a\mapsto\pair{a}{a}$.
\end{itemize}
\end{exa}

Orbit-finite sets, although usually infinite, can be presented by finite means and are therefore amenable to algorithmic manipulation. There are a few ways to do this. In~\cite{lmcs14} (with the idea going back to~\cite{CM10}), it was observed that every single-orbit equivariant set is in an equivariant bijection with a set of $k$-tuples of distinct atoms, suitably quotiented by a subgroup $G\leq \text{Sym}(k)$ of the symmetric group on $k$ elements. Thus such a set can be presented by a number $k$ and the finite group $G$, and an orbit-finite set is a formal union of such presentations. A somewhat more readable and concise scheme was used in~\cite{lics15}, where orbit-finite sets are presented by set-builder expressions of the form
\begin{align}\label{eq:set-repr}
	\{e \mid v_1,\ldots,v_n\in\Atoms, \phi\}
\end{align}
where $e$ is again an expression, $v_i$ are bound atom variables and $\phi$ is a first-order formula 
with equality.
We refer to~\cite{lics15} for a precise formulation (and to~\cite{JO16} for a proof that all orbit-finite sets can be presented this way, up to an equivariant bijection); suffice it to say that the expressions in Example~\ref{exa:sets-with-atoms} are of this form, and other similar expressions are allowed. The set defined by such an expression is supported by the atoms that appear freely in the expression.
\begin{rem}\label{rem:computability}
Using this representation, several basic operations on orbit-finite sets are computable, including:
\begin{itemize}
\item union and intersection of sets, cartesian product, set difference,
\item checking whether two sets are equal, or whether one is an element (or a subset) of the other,
\item applying an orbit-finite function to an argument, composing functions or relations, 
\item finding the direct image of a set along a function, finding the image of a subset along a relation, etc.
\item checking whether a finite set $S$ supports a given set, calculating the least support of a set, 
\item partitioning a given set into $S$-orbits, calculating the $S$-orbit of a given element.
\end{itemize}
\end{rem}
These basic operations have been implemented as components of atomic programming languages~\cite{nlambda,lois}.

The following easy but useful fact says that there are not too many orbit-finite sets or finitely supported relations (or functions) between them.

\begin{prop}\label{prop:finite-equiv}
For any finite $S\subseteq T\subseteq \Atoms$,
\begin{enumerate}
\item For any $n\in\mathbb{N}$, there are only finitely many (up to equivariant bijection) $n$-orbit sets supported by $S$.
\item For any orbit-finite set $X$ supported by $S$, there are only finitely many subsets of $X$ supported by $T$.
\item For any orbit-finite sets $X,Y$ supported by $S$, there are only finitely many relations between $X$ and $Y$ supported by $T$.
\end{enumerate}
\end{prop}
\begin{proof}
Part (1) follows from a stronger representation result mentioned above~\cite{lmcs14}, which says that every single-orbit set is a quotient of the set of all distinct tuples of atoms of a fixed length. For $n>1$, every $n$-orbit set is obviously just an $n$-tuple of single-orbit sets.

For part (2), notice that $X$ has only finitely many $T$-orbits, and a subset of $X$ supported by $T$ must arise by selecting a subset of these orbits.

Part (3) follows directly from part (2) as if $X$ and $Y$ are supported by $S$, then $X\times Y$ is supported by $S$.
\end{proof}

\subsection{Ordered atoms}\label{sec:ordered-atoms}

Intuitively, orbit-finite structures with atoms are those that can be built of atoms using equality as the only relation between them. As advocated in~\cite{lmcs14}, it is possible and useful to consider atoms with more structure. Perhaps the most significant case is that of atoms being rational numbers with the ordering relation.

Technically, this amounts to taking $\Atoms=\mathbb{Q}$ and putting $\aut(\Atoms)$ to be the group of automorphism of the total order $\mathbb{Q}$, i.e., the group of monotone bijections on rational numbers. All other definitions follow without further change. The basic results mentioned above for equality atoms hold here as well, although the fact that the least support $\supp(x)$ exists, requires a different proof (see~\cite[Cor.~9.5]{lmcs14}). In particular, operations listed in Remark~\ref{rem:computability} remain computable.

Since $\aut(\Atoms)$ is now a subgroup of the automorphism group considered for equality atoms, it is easy to see that all legal sets with equality atoms remain legal for ordered atoms. However, new legal sets, functions and relations appear.

\begin{exa}
\begin{itemize}
\item For any $a<b\in\Atoms$, the open interval $(a;b)\subseteq\Atoms$ becomes legal, at it is finitely supported by $S=\{a,b\}$. The interval forms a single $S$-orbit. The closed interval $[a;b]$ is also supported by $S$, and it comprises three $S$-orbits: $\{a\}$, $\{b\}$ and the open interval. In general, finitely supported subsets of $\Atoms$ are exactly those that are finite unions of open intervals and single points.
\item $\Atoms$ and ${\Atoms}\choose{2}$ remain equivariant single-orbit sets, but the set of ordered pairs $\Atoms^2$ now decomposes into three orbits:
\[
	\Atoms^2 = \{\pair{a}{a}\mid a\in\Atoms\} \cup \{\pair{a}{b}\mid a< b\in\Atoms\}  \cup \{\pair{a}{b}\mid a> b\in\Atoms\}.
\]
Nevertheless, $\Atoms^n$ remains orbit-finite for every $n$.
\item 
There are now two equivariant functions from ${\Atoms}\choose{2}$ to $\Atoms$: the minimum and maximum. The only equivariant function from $\Atoms$ to $\Atoms$ is still identity. There are four equivariant functions from $\Atoms^2$ to $\Atoms$: the two projections, minimum and maximum, but the only equivariant function from $\Atoms$ to $\Atoms^2$ remains the diagonal $a\mapsto\pair{a}{a}$.
\end{itemize}
\end{exa}

\noindent
Remark~\ref{rem:computability} and Proposition~\ref{prop:finite-equiv} remain true for ordered atoms.

\section{$\mu$-calculus with atoms}\label{sect_mu_atoms}
In this section we present our central definition: an extension of the classical $\mu$-calculus to sets with atoms. We begin by considering the scalar calculus; the vectorial calculus is defined towards the end of this section.

\subsection{Syntax}

Syntactically, the (scalar) $\mu$-calculus with atoms (or {\em atomic $\mu$-calculus}) is simply an extension of the (scalar) classical formalism with orbit-finite propositional connectives.

\begin{defin}
Let $\mathbb{P}$ be an equivariant set with atoms of basic propositions, and let $\mathbb{X}$ be an equivariant set with atoms of variables. The set $\aform$ of formulas of the atomic $\mu$-calculus is generated by the following grammar:
\[
    \phi::= p \mid X \mid \bigvee{\Phi} \mid \neg \phi\mid \diam\phi\mid \mu X.\phi 
\]
where $p$ ranges over $\mathbb{P}$, $X$ ranges over $\mathbb{X}$ and $\Phi$ ranges over \emph{orbit-finite} sets of formulas. As usual we only allow $\mu X.\phi$ where $X$ occurs only positively in $\phi$.  
\end{defin} 

The above definition is parametrised by a basic structure $\Atoms$ of atoms, and in particular it makes sense for both equality atoms and ordered atoms. When we need to specify and distinguish between these two cases, we will denote the equality-atomic $\mu$-calculus by $\eqaform$, and the ordered-atomic $\mu$-calculus by $\ordaform$.

The ``orbit-finite set of formulas'' in the definition refers to a canonical action of $\aut(\Atoms)$ on formulas, extending the action on $\mathbb{P}$ (and on $\mathbb{X}$, if it is nontrivial) inductively. Note that, despite ostensibly infinite disjunctions, every formula in $\aform$ has a finite depth, since two formulas in the same orbit necessarily have the same depth. Thanks to this, no need arises for transfinite induction in reasoning about the syntax of $\aform$ formulas.

We use the same syntactic conventions as in the classical case:
 \begin{multicols}{2}
\begin{itemize}
\item $ \top := p\lor\neg p$,
\item $\bigwedge \Phi := \neg\bigvee\set{\neg\phi}{\phi\in\Phi}$,
\item $\Box\phi := \neg\diam\neg\phi$,
\item $\nu X. \phi := \neg\mu X.\neg\phi[X:=\neg X]$.
\end{itemize}
\end{multicols}
With these conventions, every formula can be written in the negation normal form, where negation occurs only in front of a basic proposition or a variable. 

Often it is notationally convenient to view an orbit-finite set $\Phi$ as a family of formulas indexed by a simpler orbit-finite set. For example, we may write
\[\textstyle{
	\bigvee_{a\in\Atoms}\diam a \qquad\text{to mean}\qquad \bigvee\{\diam a\mid a\in\Atoms\}.
}\]

This quantifier-like notation suggests that orbit-finite disjunctions and conjunctions can be used to quantify over atoms existentially and universally. This intuition will be useful in reading our example formulas in what follows.

\begin{exa}\label{exa:atomic-mu-form}
Put $\mathbb{P} = \Atoms$. For every $a\in \Atoms$ let $\Phi_a = \set{\neg b}{b\in \Atoms\setminus a}$. Then $\Phi_a$ is supported by $\{a\}$ and orbit-finite, hence $\bigwedge \Phi_a:= \bigwedge_{b\neq a}\neg b$ is a formula in $\aform$. The set
		\[\textstyle{
		\Psi = \left\{\diam\left(a\wedge \bigwedge_{b\neq a}\neg b\right)\mid a\in \Atoms\right\}
		}\]
		is equivariant and orbit-finite, hence 
		\begin{equation*}
		\textstyle{\bigwedge\Psi = \bigwedge_{a\in A}\diam\left(a\wedge \bigwedge_{b\neq a}\neg b\right)
		}\end{equation*}
	        is also a legal formula in $\aform$.
\end{exa}

\begin{rem}\label{rem_klin_form}
	In the classical $\mu$-calculus, one often wants a formula to be {\em clean}, or {\em well-named}, meaning that every bound variable is bound only once. This is to ensure that to every variable one can associate a unique binding occurrence of it. In the presence of infinitary connectives, achieving cleanness may seem problematic. For example, in the formula
	\[\textstyle{
	\bigwedge_{a\in\Atoms}\left(\mu X.a\lor\diam X\right)
	}\]
the variable $X$ occurs infinitely many times, and naively replacing each binding occurrence with a completely fresh variable would result in an orbit-infinite conjunction.
This is why we will not insist that our formulas be presented in a clean form; instead, we will only require that no variable is bound more than once on any syntactic path from the root of a formula to its subformula. It is easy to rewrite any formula to an equivalent one of this form (for example by annotating each variable with a natural number that is the depth of its binding occurrence in the syntactic tree of a formula), and it will allow us to speak of ``the binding occurrence'' for any occurrence of a bound variable.
\end{rem}

The following syntactic properties can be easily proved by induction on the depth of formulas:
	\begin{itemize}
		\item Every formula $\phi$ is finitely supported.
		\item For every formula $\phi$, the set of its subformulas is finitely supported and orbit-finite.
		\item The relation ${\subf}\subseteq \aform\times\aform$ of being a subformula is equivariant, i.e. for every $\pi\in \aut(\Atoms)$ and every $\phi,\psi\in\aform$ we have
		\[\phi\subf\psi \textnormal{ if and only if } \phi\cdot\pi\subf\psi\cdot\pi.\]
	\end{itemize}

   \subsection{Semantics}
    
 Semantics of the atomic $\mu$-calculus is a straightforward extension of the classical one: we simply require all sets and relations to be sets with atoms, and replace finite models with orbit-finite ones. The following is a simple rewriting of Definition~\ref{def_krip_mod}:
	
 \begin{defin}[Kripke model]\label{def:atom_krip_mod}
An \textit{atomic Kripke model} is a triple 
\[
	\K = \trip{K}{\step^{\K}}{\models^{\K}}
\]
 such that
\begin{itemize}
\item $K$ is a set with atoms,
\item ${\step^{\K}}\subseteq K\times K$ is a finitely supported {\em transition relation},
\item ${\models^{\K}}\subseteq K\times \mathbb{P}$ is a finitely supported {\em satisfaction relation} for basic predicates.
\end{itemize}
We shall say that an atomic Kripke model $\K$ is \emph{orbit-finite} if the set $K$ is so. A set $S$ supports $\K$ if and only if it supports $K$, $\step^{\K}$ and $\models^{\K}$.
\end{defin}

We use the same notational conventions regarding dropping the superscripts as for the classical case. Also as in the classical case, the satisfaction relation can be equivalently defined by the function:
\[
	\pred(x) = \{p\in\mathbb{P} \mid x\models p\}.
\]
which is now a function from $K$ to $\mathcal{P}_{\text{fs}}(\mathbb{P})$, the set of finitely supported subsets of $\mathbb{P}$.

\begin{exa}\label{exa:krip}
A model $\mathcal{K}$ with:
\begin{itemize}
\item $\was{\star}\cup\Atoms$ as the set of states,
\item transitions $\star\step a$ for all $a\in\Atoms$,
\item satisfaction relation $\models$ such that $\pred(\star)=\emptyset$ and $\pred(a)=\{a\}$  for $a\in\Atoms$,
\end{itemize}			
is an infinite but orbit-finite, equivariant atomic Kripke model. It can be drawn as:
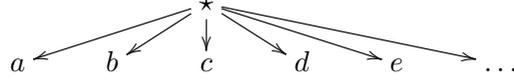
\[
			\xymatrix@R=10pt{ & & \star\ar[dll]\ar[dl]\ar[d]\ar[dr]\ar[drr]\ar[drrr] & & & \\
				a & b & c & d & e &\ldots}
\]
	\end{exa}

\begin{exa}\label{exa:krip2}
A model $\mathcal{K}$ with:
\begin{itemize}
\item $\Atoms^*$, i.e., finite sequences of atoms as states,
\item transitions $w\step wa$ for all $w\in\Atoms^*$ and $a\in\Atoms$,
\item satisfaction relation $\models$ such that $\pred(\epsilon)=\emptyset$ and $\pred(a_1\cdots a_n)=\{a_n\}$,
\end{itemize}			
is an equivariant atomic Kripke model with infinitely many orbits.
\end{exa}

	\begin{defin}\label{def:amu-sem}
	For an atomic Kripke model $\mathcal{K}$, the meaning of a formula $\phi\in\aform$ in a variable environment $\rho:\mathbb{X}\rightharpoonup\mathcal{P}(K)$ is defined exactly as in Definition~\ref{def_speln_stand}, with the obvious modification:
\begin{equation}\label{eq:orbfindis-sem}
	\left[\!\left[\bigvee \Phi\right]\!\right]_{\rho} = \bigcup \{\sembr{\phi}_{\rho}\mid\phi\in\Phi\}.
\end{equation}
	\end{defin}
	
\begin{rem}\label{rem:ZFvsZFA}
There is an apparent lack of consistency in literally transporting Definition~\ref{def_speln_stand} to sets with atoms. In one clause taken from that definition:
\[
    \sembr{\mu X.\phi}_\rho = \mathrm{lfp}(F), \text{where } F(A)=\sembr{\phi}_{\rho[X\mapsto A]},
\]
the least fixpoint of a function $F:{\mathcal P}(K)\to {\mathcal P}(K)$ is considered. However, as mentioned in Example~\ref{exa:sets-with-atoms}, ${\mathcal P}(K)$ is not a legal set with atoms in general. The set $\mathcal{P}_{\text{fs}}(K)$ of finitely supported subsets of $K$ is legal, and one could consider a restriction of $F$ to:
\[
    F: \mathcal{P}_{\text{fs}}(K) \to \mathcal{P}_{\text{fs}}(K).
\]
However, $\mathcal{P}_{\text{fs}}(K)$ is not a complete lattice, so Tarski's theorem does not immediately apply to it. 
	
Our approach is to interpret Definition~\ref{def_speln_stand} and its extending clause in Definition~\ref{def:amu-sem} literally, initially without regard for the legality of sets with atoms, and then to prove Lemma~\ref{lem:semfinsupp} below, which shows by induction that the semantics of each formula, including fixpoint formulas, is in fact a legal, finitely supported set.
 
A more principled approach would be to restrict functions $F$ and environments $\rho$ to finitely supported sets from the beginning. This would amount to reinterpreting Definition~\ref{def_speln_stand} in the set theory with atoms ZFA, an approach pursued e.g.~in~\cite{gabbay-pitts}. One then needs to show a suitable version of Tarski's theorem to infer the existence of least fixpoints. Such a theorem indeed holds in a general topos-theoretic setting (see e.g.~\cite{bauer13}); for sets with atoms, it guarantees the existence of least fixpoints of monotone, finitely supported functions on lattices where least upper bounds exist for all finitely supported families (see also e.g.~\cite{LoschPitts}).

We choose the more elementary approach to avoid a formal reliance on ZFA. We then need to state the following lemma, which implies that the semantics of every formula is finitely supported.
\end{rem} 	
    
\begin{lem}\label{lem:semfinsupp}
For every $\phi\in\aform$, an atomic Kripke model $\mathcal{K}$ and a variable environment $\rho$, the set $\sembr{\phi}_{\rho}\subseteq K$ is supported by
$S=\supp(\phi)\cup\supp({\K})\cup B\subseteq\Atoms$,
where $B$ is any support of $\rho$.
\end{lem}
\begin{proof}[Proof (sketch)]
By structural induction on $\phi$. For example, consider the case of orbit-finite disjunction above, i.e., $\phi=\bigvee\Phi$. By the inductive assumption, the set $\{\sembr{\phi}_{\rho}\mid\phi\in\Phi\}$ is supported by $S$. Since $\bigcup$ is an equivariant operation on finitely supported families of sets with atoms, the lemma follows. 

The most interesting case is that of the fixpoint operator. Recall that
\[
\sembr{\mu X.\phi}_\rho = \mathrm{lfp}(F), \text{where } F(A)=\sembr{\phi}_{\rho[X\mapsto A]}
\]
is obtained as the union of an increasing chain of approximants:
\[
    A^0\subseteq A^1\subseteq A^2 \subseteq \cdots A^\omega\subseteq A^{\omega+1}\subseteq \cdots
\]
defined as in Definition~\ref{def_speln_stand}. By ordinal induction, and using the inductive assumption about $\phi$, 
it is easy to show that $A^\alpha$ is supported by $S$ for every ordinal $\alpha$. Indeed, if $A^\alpha$ is supported by $S$ then for any $\pi\in\aut_S(\Atoms)$:
\[
    A^{\alpha+1}\cdot\pi = (\sembr{\phi}_{\rho[X\mapsto A^{\alpha}]}) \cdot\pi =  \sembr{\phi}_{\rho[X\mapsto A^{\alpha}]} = A^{\alpha+1}
\]
where the middle equality holds by the inductive assumption about $\phi$, using the simple observation that, since $S$ supports $\mu X.\phi$ (hence also $\phi$ and $X$), $\rho$ and $A^\alpha$, then it also supports $\rho[X\mapsto A^{\alpha}]$.

As a result, the set $\sembr{\mu X.\phi}_\rho$, being the union of an increasing chain of $S$-supported sets, it itself $S$-supported.
\end{proof}


Atomic CTL$^*$ and atomic LTL are defined by analogy to atomic $\mu$-calculus, extending Definition~\ref{def:ctl} with orbit-finite disjunctions, with semantics extended by analogy to~\eqref{eq:orbfindis-sem} in Definition~\ref{def:amu-sem}. With CTL$^*$ a design decision is to be made: in the semantic clause
\begin{itemize}
\item ${x}\models \exists\phi\iff$ for some path $\pi$ starting at $x, \pi\models\phi$,
\end{itemize}
do we require the path $\pi$ to be finitely supported or not? Note that even an equivariant Kripke model can contain infinite paths that are not finitely supported. As it turns out, this choice matters little: in Theorem~\ref{thm:ctl-mc-undec} we shall prove that model checking for atomic CTL$^*$ is undecidable and the argument there works whether or not we choose paths as being finitely supported.


\subsection{The vectorial calculus}

A beginner's recipe for moving atom-less mathematical notions to sets with atoms is rather straightforward: replace all functions and relations by equivariant or finitely supported ones, all finite structures by orbit-finite ones, and then see what happens. With this in mind, notice that the classical definition of vectorial $\mu$-calculus includes a finiteness condition that is not explicitly present in the scalar $\mu$-calculus: namely, in a vectorial formula as in~\eqref{eq:vectorial}, the set of variables $X_j$ is finite. A natural idea is therefore to replace that set with an orbit-finite one.

Formally, in the {\em vectorial atomic $\mu$-calculus} $\vaform$, the fixpoint construction is generalized to the form
\[
	\phi = \mu \check{X}.\Phi
\]
where $\check{X}\in\mathbb{X}$ is the {\em entry variable} and $\Phi$ is an orbit-finite set of \emph{equations} of the form
$X.\phi$,
$X\in \mathbb{X}$ and $\phi$ is a formula. We assume that in $\Phi$ no variable occurs on the left-hand side of two different equations, and that the entry variable occurs on the left-hand side of some equation. 

As with orbit-finite boolean connectives, it is often useful to see the orbit-finite set $\Phi$ as a family of equations
\[
\left\{X_j.\phi_j\right\}_{j\in J}
\]
indexed by some simpler orbit-finite set $J$, with the entry variable pointed to by some $i\in J$. The formula $\phi$ can be then written as
\[
	\phi = \mu X_i.\left\{X_j.\phi_j\right\}_{j\in J},
\]
quite resembling the classical vectorial formula~\eqref{eq:vectorial2}.

The semantics of such a formula is defined very much like in the classical vectorial calculus, except that here the operator whose fixpoint is calculated runs not on finite tuples $(A_1,\ldots,A_n)$, but tuples $(A_j)_{j\in J}$, indexed by an orbit-finite set $J$, where $A_j\subseteq K$ (considerations in Remark~\ref{rem:ZFvsZFA} apply here as well). The set of such tuples is (in a classical sense) a complete lattice under pointwise intersections and unions, so least fixpoints of monotone operators exist on it. If the operator defined by:
\[
	\left(F(A_j)_{j\in J}\right)_k = \sembr{\phi_k}_{\rho[X_j\mapsto A_j]_{j\in J}} \quad \text{for } k\in J
\]
has a tuple $(B_j)_{j\in J}$ as its least fixpoint, we put
\[
	\sembr{\phi} = B_i
\]
(where $X_i$ is the entry variable for $\phi$) as in the classical setting.

It is easy to check that Lemma~\ref{lem:semfinsupp} still holds for this extended semantics, with essentially the same proof.

This syntax and semantics makes sense for any choice of atoms. When we need to specify and distinguish between the two specific cases considered in this paper, we will denote the vectorial equality-atomic $\mu$-calculus by $\veqaform$, and the vectorial ordered-atomic $\mu$-calculus by $\vordaform$.

\begin{exa}
Working over ordered atoms, consider a vocabulary of basic propositions that coincides with the set $\Atoms$ of atoms. Let us analyse the following formula in $\vordaform$:
\[
	\phi = \bigvee_{a\in\Atoms}\nu X_a.\left\{X_b.(\diam b)\land \bigvee_{c>b}X_c\right\}_{b\in\Atoms}.
\]
Here, the indexing set $J=\Atoms$ is single-orbit, and for each atom $b$, the formula
\[	
	\phi_b = (\diam b)\land \bigvee_{c>b}X_c
\]
says that:
\begin{itemize}
\item some immediate successor of the current state satisfies $b$, and
\item for some $c>b$, the property $X_c$ (which, in the fixpoint construction, will be bound to an analogous formula $\phi_c$) holds in the current state.
\end{itemize}
Altogether, the formula $\phi$ says that there is an infinite increasing chain of atoms such that each atom in that chain is satisfied in some immediate successor of the current state. This, by the way, is equivalent to saying that there is any infinite {\em set} of atoms with this property; this is because the set of immediate successors of a particular state in a Kripke model is always finitely supported, and every infinite but finitely supported subset of $\Atoms$ contains an open interval and therefore it contains an infinite increasing chain.

It is instructive to see how the greatest fixpoint in the semantics of $\phi$ is approximated on an equivariant Kripke model $\K$, by successive applications of a monotone operator on $\Atoms$-indexed tuples of subsets of $K$:
\begin{itemize}
\item at the starting point, every variable $X_b$ is mapped to the full set $K$;
\item after one step, $X_b$ is mapped to the set all states where some successor satisfies $b$;
\item after two steps, $X_b$ is mapped to the set of all states where:
\begin{itemize}
\item some successor satisfies $b$,
\item some successor satisfies some atom greater than $b$;
\end{itemize}  
\item $\ldots$
\item after $k$ steps, $X_b$ is mapped to the set of all states where:
\begin{itemize}
\item some successor satisfies $b$,
\item there are at least $k-1$ atoms greater than $b$ that are satisfied in some successor.
\end{itemize}  
\end{itemize}
Note that, after each step, the approximate interpretation of $X_b$ is a subset of $K$ supported by $\{b\}$; moreover, the function that maps each $X_b$ to its approximate interpretation is equivariant. From this it follows that if the model $\K$ is orbit-finite then the fixpoint is reached after finitely many steps. 
\end{exa}

In the atom-less world one does not gain expressivity by allowing fixpoints over systems of equations. This is a consequence of the following Beki\'c principle:

\begin{thm}[Beki\'c]
Let $D$, $E$ be any complete lattices and suppose $F:D\times E\rightarrow D$ and $G: D\times E\rightarrow E$ are two monotone functions. Then the least fixpoint of $\pair{F}{G}:D\times E\rightarrow D\times E$ is the pair $\pair{\widehat{f}}{\widehat{g}}$, where
\begin{align*}
\widehat{f} &= \mu f. F(f,\mu g. G(f,g))\\
\widehat{g} &= \mu g. G(\mu f. F(f,g),g)
\end{align*}
\end{thm}

As we shall see in Section~\ref{sect_separation}, contrary to the classical setting, vectorial \emph{atomic} $\mu$-calculus is strictly more expressive than its scalar counterpart. However, Beki\'c principle can still be used some extent, to obtain a more regular normal form for vectorial atomic formulas. Recall that when introducing syntax we allowed the fixpoint to be taken over arbitrary orbit-finite sets of equations. Next proposition shows that in fact we can restrict ourselves to just {\em single-orbit} systems of equations. More precisely: we can demand that in a well-built formula 
\[\mu X_i.\was{X_j.\phi_j}_{j\in J}\]
then the set $\was{X_j.\phi_j}_{j\in J}$ is single-orbit.

\begin{stw}
	Every formula of vectorial atomic $\mu$-calculus is semantically equivalent to one with fixpoints over single orbit sets of equations.
\end{stw}

\begin{proof}
It is sufficient to show that if 
\[\mu X_i.\was{X_j.\phi_j}_{j\in J}\]
is a well-built formula and $\was{X_j.\phi_j}_{j\in J}$ has two orbits, then we can split the fixpoint into two nested fixpoints over single-orbit sets of equations. Let 
\[\was{X_j.\phi_j}_{j\in J_1} \ \ \textnormal{ and } \ \ \was{X_j.\phi_j}_{j\in J_2}\]
be the two orbits of $\was{X_j.\phi_j}_{j\in J}$. Given an arbitrary Kripke model $\K$ we can see them as monotone functions
\begin{align*}
F_1: \power(K)^{J_1}\times \power(K)^{J_2} & \rightarrow \power(K)^{J_1}\\
F_2: \power(K)^{J_1}\times \power(K)^{J_2} & \rightarrow \power(K)^{J_2}
\end{align*}
so that the function defined by the whole system $\was{X_j.\phi_j}_{j\in J}$ can be presented as $\pair{F_1}{F_2}$ (note that $\power(K)^{J_1}\times \power(K)^{J_2}$ is canonically isomorphic to $\power(K)^{J}$). By the Beki\'c principle the least fixpoint of $\pair{F_1}{F_2}$ is 
\begin{align*}
\mu f. F_1(f,\mu g. F_2(f,g)),\\
\mu g. F_2(\mu f. F_1(f,g),g).
\end{align*}
Without loss of generality suppose that $X_i \in J_1$. Then it follows that $\mu X_i.\was{X_j.\phi_j}_{j\in J}$ is equivalent to 
\[\mu X_i.\was{X_j.\phi'_j}_{j\in J_1},\]
where 
\[\phi'_j := \phi_j[X_k\mapsto \mu X_k.\was{X_l.\phi_l}_{l\in J_2}]_{k\in J_2}.
\]
Note that $\mu X_k.\was{X_l.\phi_l}_{l\in J_2}$ is by definition the $k$-th projection of the least fixpoint of $\was{X_j.\phi_j}_{j\in J_2}$; $\phi_j[X_k\mapsto \theta_k]_{k\in J_2}$ denotes the simultaneous substitution of all formulas $\theta_k$ for the respective variables in $\phi_j$.
\end{proof}


\section{Model checking}\label{sect_model_checking}
As we shall see in the following sections, the landscape of modal $\mu$-calculi with atoms is considerably more complex than the classical setting. Some results that hold classically fail here, some others behave differently depending on the chosen structure of atoms. However, some important results from the classical calculus transport to the atomic setting. We begin by proving that this is the case for the decidability of model checking. We state the theorem for the vectorial calculus, and the scalar counterpart follows trivially.

\begin{thm}\label{thm:mc-decid}
	Model-checking problem for the vectorial atomic $\mu$-calculus $\vaform$ is decidable over orbit-finite atomic Kripke models.
\end{thm}
\begin{proof}
Let us fix an orbit-finite atomic Kripke model $\mathcal{K}$. We shall show that the meaning of any formula $\phi\in\vaform$ in $\mathcal{K}$ (under a variable environment $\rho$) can be computed from $\phi$, by structural induction on $\phi$ and using basic operations listed in Remark~\ref{rem:computability}.

In the inductive computation, we shall only ever consider environments $\rho$ that are defined only on an orbit-finite set of variables, and are supported by (the support of $\K$ together with) the support of some subformula of $\phi$. Every such environment is presentable by finite means and is amenable to algorithmic manipulation.

The cases of basic propositions, variables, negation and the modality $\diam$ are straightforward.
For the case of orbit-finite disjunction $\phi=\bigvee\Phi$, first calculate 
\[
S=\supp(\Phi)\cup\supp(\mathcal{K})\cup\supp(\rho).
\]
Then partition $\Phi$ into $S$-orbits (there are finitely many of them), and select a system of representatives $\phi_1,\ldots,\phi_n$, one from each orbit. Using the inductive assumption, calculate $P_i=\sembr{\phi_i}_{\rho}$ for each $i$. Then compute the $S$-orbit $\mathcal{O}_i$ of  each $P_i$; each $\mathcal{O}_i$ is an $S$-supported family of subsets of $K$. The union of all $\bigcup\mathcal{O}_i$ is the desired set $\sembr{\phi}_\rho$.
        
The most interesting case is computing 
\[
	\sembr{\mu X_i.\left\{X_j.\phi_j\right\}_{j\in J}}_\rho.
\] 
This is done by approximating the least fixpoint by the following standard procedure:
\begin{itemize}
\item[(1)] Put $A_j=\emptyset$ for each $j\in J$;
\item[(2)] Extend $\rho$ by mapping each variable $X_j$ to $A_j$;
\item[(3)] Using the inductive assumption, calculate $\sembr{\phi_j}_{\rho[X_k\mapsto A_k]_{k\in J}}$ for each $j\in J$, and assign the resulting tuple as a new value of $A$;
\item[(4)] Repeat steps (2)-(3) until $A$ stabilizes;
\item[(5)] Return $A_i$.
\end{itemize}
Step (3) above is achieved by choosing an arbitrary representative $j$ from each $S$-orbit of $J$, computing the subset $\sembr{\phi_j}_{\rho[X_k\mapsto A_k]_{k\in J}}\subseteq K$ for each representative, and (uniquely) extending the result to an $S$-supported relation between $J$ and subsets of $K$. 

The set of all $J$-indexed tuples of subsets of $K$ is a complete lattice, so by the Knaster-Tarski 
theorem (see \cite{niw-mu}) all we need to show is that this procedure terminates. Notice that, at every stage of computation, the current value of $(A_j)_{j\in J}$ defines a relation between $J$ and $K$. Moreover, by Lemma~\ref{lem:semfinsupp} applied to the vectorial calculus, each such relation is finitely supported by 
\[
S=\supp(\phi)\cup\supp(\mathcal{K})\cup\supp(\rho).
\] 
Since $K$ is orbit-finite, by Proposition~\ref{prop:finite-equiv}(3) $A$ can take on only finitely many values, therefore the above procedure terminates after finitely many steps.
\end{proof}


\section{Failure of orbit-finite model property}\label{sect_finite_model_property}
The classical modal $\mu$-calculus enjoys the so-called finite model property: every satisfiable formula has a finite model. (In fact a stronger {\em small model property} holds, useful for complexity upper bounds.) There is no chance for this property to hold in the atomic setting, but since orbit-finite sets play the role of finite sets in the universe of sets with atoms, one might hope that an {\em orbit-finite model property} holds, i.e., that every satisfiable formula in $\aform$ has an orbit-finite model.

However, even that weaker property fails even for the scalar $\mu$-calculus $\aform$ both for equality and ordered atoms, as we shall now prove. We apply different arguments for the two atom structures.

\subsection{The case of equality atoms}

Over a vocabulary of basic predicates that includes a predicate for each atom $a$, consider the following two properties:
\begin{enumerate}[label=\bf P\arabic*:]
\item every reachable state has at least one successor for which some basic predicate holds;
\item on every path that starts from the current state, no basic predicate holds more than once.
\end{enumerate}

These are definable in the atomic $\mu$-calculus $\eqaform$:
\begin{align*}
 \text{{\bf P1}\quad is \quad} &\nu X.\left(\left(\diam\bigvee_{a\in\Atoms}a\right)\land\Box X\right), \\
\text{{\bf P2}\quad is \quad} &\neg(\mu X.(\psi\lor \diam X)) \qquad\text{where}\qquad \psi = \bigvee_{a\in\Atoms}(a\land\diam\mu Y.(a\lor\diam Y)).
\end{align*}

The conjunction of {\bf P1} and {\bf P2}  is satisfiable. Indeed, it is satisfied in (every state) of the Kripke model from Example~\ref{exa:krip2} restricted to those states $w\in\Atoms^*$ where $w$ does not contain any letter more than once. It is easy to see that this is a well-defined equivariant model, and that {\bf P1} and {\bf P2} hold in it. Note that this model still has infinitely many orbits of states, with exactly one orbit for each length of $w$.

On the other hand, the conjunction of {\bf P1} and {\bf P2} has no orbit-finite models. Indeed, assume that some state $x_0$ in such a model  satisfies both properties. By {\bf P1}, there exists an infinite path in the model:
\[
	x_0 \step x_1 \step x_2 \step  x_3 \step x_4 \step \cdots.
\]
where each state (perhaps except $x_0$) satisfies some basic predicate. By {\bf P2} no such predicate is satisfied more than once on this path.

Since the model is orbit-finite, there exists a global upper bound on the size of the least supports $\supp(x_i)$. This implies that there exists a number $j$ and distinct atoms $a,b$ such that:
\begin{itemize} 
\item $a$ holds in some $x_i$ where $i<j$,
\item $b$ holds in some $x_k$ where $j<k$,  and
\item $a,b\not\in\supp(x_j)$.
\end{itemize}
Let $\pi\in\aut(\Atoms)$ be the atom automorphism that swaps $a$ and $b$ and leaves all other atoms untouched; then $x_j\cdot\pi=x_j$, therefore ${x_{j-1}}\step {x_j\cdot\pi}$ is a valid transition. As a result:
\[
	x_0\step x_1\step \cdots\step x_i\step\cdots\step x_{j-1}\step x_j\cdot\pi\step \cdots\step x_k\cdot\pi\step \cdots
\]
is a legal path in the model. But $a$ holds both in $x_i$ and in $x_k\cdot\pi$, so {\bf P2} is violated on this path.

\subsection{The case of ordered atoms}

Properties {\bf P1} and {\bf P2} above make sense also for ordered atoms, and they are defined in $\ordaform$ by the same formulas as before; the ordering relation on atoms is simply not used in those formulas. The orbit-infinite model from before still satisfies both properties.

However, over ordered atoms, the above proof of the lack of an orbit-finite model fails: the bijection $\pi$ that swaps $a$ and $b$ is not a valid atom automorphism. In fact, over ordered atoms the conjunction of  {\bf P1} and {\bf P2} has an orbit-finite model. Indeed, consider a Kripke model with $\Atoms$ as the set of states, the satisfaction relation defined by $a\models a$, and the transition relation by:
\[
	a\step b \text{ if and only if } a<b.
\]
It is easy to see that both {\bf P1} and {\bf P2} hold for (every state of) this model.

Consider, however, a stronger version of {\bf P1}, defined with an essential use of the ordering relation between atoms:
\begin{enumerate}[label=\bf P\arabic*':]
\item for every open interval $(a;b)\subseteq\Atoms$, every reachable state has at least one successor for which some basic proposition $c\in(a;b)$ holds.
\end{enumerate}
This is definable in the atomic $\mu$-calculus $\ordaform$, by the formula
\[
	\bigwedge_{a<b\in\Atoms}\nu X.\left(\left(\diam\bigvee_{c\in(a;b)}c\right)\land\Box X\right).
\]
The conjunction of {\bf P1'} and {\bf P2} is satisfiable; indeed, it is satisfied in (every state of) the orbit-infinite model that we used to model {\bf P1} and {\bf P2} for equality atoms.

On the other hand, the conjunction of {\bf P1'} and {\bf P2} has no orbit-finite models over ordered atoms. Indeed, assume that some state $x_0$ in such a model  satisfies both properties. By {\bf P1'}, there exists an infinite path in the model:
\[
	x_0 \step x_1 \step x_2 \step  x_3 \step x_4 \step \cdots.
\]
where each state (perhaps except $x_0$) satisfies some basic predicate. By {\bf P2} no such predicate is satisfied more than once on this path.

Since the model is orbit-finite, there exists a global upper bound on the size of the least supports $\supp(x_i)$. This implies that there exists a number $j$ and an atom $d$ such that:
\begin{itemize} 
\item $d$ holds in some $x_i$ where $i<j$,
\item $d\not\in\supp(x_j)$.
\end{itemize}
Pick some open interval $(a;b)$ that contains $d$ but does not contain any atom from $\supp(x_j)$. By {\bf P1'}, $x_j$ has some successor (call it $y$) that satisfies some basic proposition $c\in(a;b)$. 

Let $\pi\in\aut(\Atoms)$ be an atom automorphism that maps $c$ to $d$ and leaves all atoms from $\supp(x_j)$ untouched; then $x_j\cdot\pi=x_j$, therefore ${x_{j-1}}\step {x_j\cdot\pi}$ is a valid transition. As a result:
\[
	x_0\step x_1\step \cdots\step x_i\step\cdots\step x_{j-1}\step x_j\cdot\pi\step y\cdot\pi
\]
is a legal path in the model. But $a$ holds both in $x_i$ and in $y\cdot\pi$, so {\bf P2} is violated on this path.


\section{Parity games}\label{sect_parity_games}
An important result in the theory of classical $\mu$-calculus is its correspondence to {\em parity games}. In this section we show that this correspondence extends, without significant conceptual changes, to the setting with atoms. This has a double purpose. First, since (as we prove) orbit-finite games are decidable, the correspondence gives an alternative route to deciding the model-checking problem for atomic Kripke models. Second, in Section~\ref{sect_bisimulations} we will use parity games to show that modal $\mu$-formulas are invariant under {\em atomic bisimulations}, which will be used in Sections~\ref{sect_separation}-\ref{sect_limitations} for proving expressive limitations of atomic $\mu$-calculi.

The following development closely follows analogous definitions and results for the classical, atom-less $\mu$-calculus; see e.g.~\cite{BraWal15} or~\cite{yde} for a detailed exposition of that theory.

\begin{defin}
		An \emph{atomic parity game} $\mathcal{G}$ is a quadruple $\pargame$ such that
		\begin{enumerate}
			\item $V$ is a set (with atoms) of {\em nodes}, $V_{\exists}$ is a finitely supported subset and we define $V_{\forall}= V\setminus V_{\exists}$;
			\item $R\subseteq V^2$ is a finitely supported {\em move} relation;
			\item $\Omega: V\rightarrow \nat$ is a bounded, finitely supported {\em rank} function.
		\end{enumerate}
		The game is called \emph{orbit-finite} if $V$ is orbit-finite.
	\end{defin}
In a context of a parity game, a {\em match} (starting at a node $v_0$) is a sequence, finite or infinite, of nodes:
\[
	\vec{v} = v_0,v_1,v_2,\ldots
\]
such that $\pair{v_n}{v_{n+1}}\in R$ for each $n\in\mathbb{N}$. A match is {\em complete} if it cannot be extended to a longer match, i.e., if it is infinite or if it ends with a node that has no successors in $R$. Otherwise the match is {\em partial}. A complete match $\vec{v}$ is {\em won} by $\exists$ if it is infinite and $\limsup_n \Omega(v_n)$ is an even number, or if it is finite and its last node is in $V_{\forall}$. Otherwise $\vec{v}$ is won by $\forall$. 

For a node $v_0\in V$, let $\PM_{\exists}(v_0)$ be the set of partial matches that begin in $v_0$ and end in a node from $V_{\exists}$. A {\em strategy} for $\exists$ is a function $\rho:\PM_{\exists}(v_0)\to V$ such that for each partial match $\vec{v}=\langle v_0,\ldots,v_n\rangle\in\PM_{\exists}(v_0)$ it holds that $\pair{v_n}{\rho(\vec{v})}\in R$. A strategy is {\em positional} if $\rho(v_0,\ldots,v_n)$ only depends on $v_n$;
positional strategies can be seen as functions from $V_{\exists}$ to $V$.
A match $\vec{v}=v_0,v_1,v_2,\ldots$ {\em conforms} to a strategy $\rho$ if for every $n$ such that $v_n\in V_{\exists}$ there is $v_{n+1}=\rho(v_0,\ldots,v_n)$ ($v_{n+1} = \rho(v_n)$ if $\rho$ is positional). A strategy $\rho$ is {\em winning} for $\exists$ if every match that conforms to $\rho$ is won by $\exists$. If a positional winning strategy exists for $v_0$, we say that $v_0$ is a {\em winning node} for $\exists$. The set of such nodes is denoted $\textrm{Win}_{\exists}(\mathcal{G})$.
Strategies and winning strategies for $\forall$, and the set $\textrm{Win}_{\forall}(\mathcal{G})$ of nodes winning for $\forall$, are defined analogously. 

It is a standard result that every parity game $\mathcal{G}$ is {\em positionally determined}, i.e., for every node in it either there is a positional winning strategy for $\exists$ or for $\forall$.
Atomic parity games are obviously (forgetting about the action of atom automorphisms) parity games in the classical sense, so they are positionally determined. However, it might happen that in an atomic parity game no winning strategy is finitely supported. 

\begin{exa}
Consider an atomic parity game (over equality atoms) where:
\begin{align*}
V = \textstyle{\Atoms\choose{2}}\cup\Atoms \qquad V_\exists = \textstyle{\Atoms\choose{2}} 
\qquad R = \{\pair{\{a,b\}}{a}\mid a,b\in\Atoms\} \cup \{\pair{a}{\{b,c\}}\mid a,b,c\in\Atoms\} \\
\Omega(v) =0 \text{ for all $v\in V$.}
\end{align*}
Since $\exists$ wins every infinite play and every state has a successor with respect to $R$, it is clear that every state is winning for $\exists$. However, no winning strategy for $\exists$ is finitely supported. Indeed, such a strategy would determine a finitely supported function from $\Atoms\choose{2}$ to $\Atoms$ such as $f(C)\in C$ for all $C\in{\Atoms\choose{2}}$, and it is easy to see that over equality atoms no such function exists. 
\end{exa}	

In spite of this, winning regions in orbit-finite parity games are computable. Indeed, every orbit-finite game can be effectively transformed into a finite game in the following way. For an orbit-finite parity game $\mathcal{G} = \pargame$, let $S$ be any finite set of atoms that supports $\mathcal{G}$. Let $V/S$, $V_{\exists}/S$ be the sets of $S$-orbits of $V$ and $V_{\exists}$, respectively; let $[v]$ denote the $S$-orbit of $v\in V$. Obviously $V_{\exists}/S\subseteq V/S$. Define $R/S\subseteq V/S\times V/S$ and $\Omega/S:V/S\to\mathbb{N}$ by:
		\begin{eqnarray}
		\pair{[v]}{[w]}\in R/S &\textnormal{ if }& \pair{x}{y}\in R \textnormal{ for some } x\in[v], y\in[w]\nonumber\\
		\Omega/S([v])=n &\textnormal{ if }& \Omega(x) = n \textnormal{ for some } x\in [v]\nonumber
		\end{eqnarray}
This is well defined since $S$ supports both $R$ and $\Omega$. In particular, $\Omega/S$ is a function. We call $\mathcal{G}/S=\langle V/S,V_{\exists}/S,R/S,\Omega/S\rangle$ the \emph{orbit game} of $\mathcal{G}$.

Games $\mathcal{G}$ and $\mathcal{G}/S$ are very similar. Indeed, forgetting about the winning condition, the game $\mathcal{G}$ can be seen as a Kripke model over $\mathbb{N}$ as the set of basic predicates: put $V$ as the set of states, $R$ as the transition relation and the satisfaction relation defined by
\[
	v\models n \text{ if and only if } \Omega(v)=n.
\]
The same can be said for $\mathcal{G}/S$. As it turns out, $\mathcal{G}$ and $\mathcal{G}/S$ are bisimilar in the classical sense:

\begin{lem}
The quotient function $\Pi$ defined for every $v\in V$ by $\Pi(v) = [v]$ is a bisimulation between $\mathcal{G}$ and $\mathcal{G}/S$ understood as Kripke models.
\end{lem}
		\begin{proof}
First, if $\Pi(w)=[v]$, then $w\in [v]$ and consequently all the elements of $[v]$ have label $\Omega(w)$. So, by definition $\Omega/S([v]) = \Omega(w)$. 

Now suppose $\Pi(w) = [v]$ and $\pair{[v]}{[z]}\in R/S$. It means that there are $x\in [v]$, $y\in [z]$ such that $\pair{x}{y}\in R$. Since $w$ and $x$ are in the same $S$-orbit, pick a $\pi\in\aut_S(\Atoms)$ such that $x\cdot\pi = w$. Since $S$ supports $R$, we get $\pair{w}{y\cdot \pi}\in R$. But $y\cdot \pi\in[z]$, so $\Pi(y\cdot \pi) = [z]$. 

For the opposite direction, let $\Pi(w) = [v]$ and for some $x$, $\pair{w}{x}\in R$. Then by the definition of $R/S$, $\pair{[v]}{[x]}\in R/S$ and obviously $\Pi(x) = [x]$.
		\end{proof}

Moreover, for every $v\in V$, $v\in V_{\exists}$ if and only if $[v]\in V_{\exists}/S$. As a result, $\exists$ has a (positional) winning strategy from $v$ in $\mathcal{G}$ if and only if she has a (positional) winning strategy from $[v]$ in $\mathcal{G}/S$. This implies that one can effectively decide whether a player has a winning strategy in an orbit-finite atomic parity game $\mathcal{G}$ by calculating first $S=\supp(\mathcal{G})$, then $\mathcal{G}/S$, and finally solving the analogous problem in the finite parity game obtained, using standard methods.

A correspondence of the atomic $\mu$-calculus with atomic parity games relies on defining,
for an atomic modal $\mu$-formula and a  Kripke model $\K$, a parity game with atoms $\evalg{\phi}$
whose nodes are pairs consisting of (occurrences of) subformulas of $\phi$ and states of $\K$, such that $x\models^{\K} \phi$ if and only if $\pair{\phi}{k}$ is a winning node for $\exists$ in $\evalg{\phi}.$
We shall show how to do this for the
\emph{vectorial} $\mu$-calculus $\vaform$; the analogous result for the scalar calculus $\aform$ is a special case. Both results are uncomplicated translations of the classical theorem (see~e.g.~\cite{BraWal15,yde}); the main difference is that the classical proof deals only with the scalar calculus (which in the atom-less world is sufficient, because the vectorial calculus can be translated into the scalar one).

First, notice that alternation depth of a formula can be defined exactly as in the classical case (see e.g. \cite{BraWal15}). In particular, if 
\[\was{X_j.\psi_j}_{j\in J}\]
is an orbite-finite system of equations, then $\psi_j$ are of finitely many syntactical shapes, hence $\mu X_i. \was{X_j.\psi_j}_{j\in J}$ can be assigned a depth as in the atom-less case. The alternation depth of $\phi$ will be denoted by $\alpha(\phi)$. 

We work with a fixed model $\K$, context $\rho$ and a fixed formula $\phi$ in negation normal form, possibly containing free fixpoint variables. 
The nodes of $\evalg{\phi}$ are pairs of the form $\pair{\psi}{x}$
where $x$ is a state in $\K$ and $\psi$ is an \emph{occurrence} of a subformula of $\phi$ (we do not count equations and systems of equations as subformulas; for the sake of simplicity we shall be talking about formulas and not their occurrences, but it must be noted that distinguishing between the two is important for the parity game to be well-defined). This set of nodes is orbit-finite and finitely supported by $\supp(\phi)\cup\supp(\K)$.

Nodes {\em controlled by $\exists$} are of the following types:
\begin{itemize}
\item $\pair{X}{x}$, where $X$ is a free variable in $\phi$ and $x\notin \rho(X)$;
\item $\pair{p}{x}$, where $p$ is a basic predicate, and $x\notin \sembr{p}_{\rho}$\ ;
\item $\pair{\neg p}{x}$, where $p$ is a basic predicate and $x\in \sembr{p}_{\rho}$\ ;
\item $\pair{\bigvee\Phi}{x}$;
\item $\pair{\diam \psi}{x}$.
\end{itemize}

Dually, nodes {\em controlled by $\forall$} are of the following form
\begin{itemize}
\item $\pair{X}{x}$, where $X$ is a free variable in $\phi$ and $x\in \rho(X)$;
\item $\pair{p}{x}$, where $p$ is a basic predicate, and $x\in \sembr{p}_{\rho}$\ ;
\item $\pair{\neg p}{x}$, where $p$ is a basic predicate and $x\notin \sembr{p}_{\rho}$\ ;
\item $\pair{\bigwedge\Phi}{x}$;
\item $\pair{\Box \psi}{x}$.
\end{itemize}
Nodes of the form $\pair{\mu X.\Phi}{x}$, $\pair{\nu X.\Phi}{x}$ and $\pair{X}{x}$ where $X$ is a bound variable, are not controlled by either player. This is not a problem since from each such node there will be only one move available. 

Depending on a type of a node, the following {\em moves} are available:
\begin{description}
\item[\textbf{No move}] from a node of type $\pair{X}{x}$, where $X$ is a free variable, or nodes of type $\pair{p}{x}$ or $\pair{\neg p}{x}$ where $p$ is a basic predicate, no moves are available.
\item[\textbf{Modal move}] from a node of type $\pair{\diam\phi}{x}$ or $\pair{\Box\phi}{x}$ the player that controls the node can move to any node $\pair{\phi}{y}$ such that $x\longrightarrow^{\K} y$;
\item[\textbf{Boolean move}] from a node of type $\pair{\bigwedge \Phi}{x}$ or $\pair{\bigvee \Phi}{x}$, the controlling player can move to any node $\pair{\phi}{x}$ such that $\phi\in \Phi$;
\item[\textbf{Automatic move}] there are two types of automatic moves:
\begin{itemize}
\item from a node of type $\pair{\mu X.\Phi}{x}$ or $\pair{\nu X.\Phi}{x}$ the game moves to the node $\pair{\psi}{x}$ where $\psi$ is the unique formula such that $X.\psi\in\Phi$;
\item from a node of type $\pair{X}{x}$ where $X$ is a variable bound by a (unique) subformula $\mu X'.\Phi$ (or $\nu X'.\Phi$) of $\phi$, the game moves to the node $\pair{\psi}{x}$ where $\psi$ is the unique formula such that the equation $X.\psi$ is in $\Phi$.
\end{itemize}
\end{description}

Modal and boolean moves above are essentially as in the classical setting. The only substantial difference is for automatic moves: the idea is that players are allowed to (and indeed are forced to) travel across various equations from a given system. When a node with a bound variable $X$ is reached, the game automatically transfers to the $X$-th equation of the system that binds it.

Nonzero {\em ranks} are assigned only to nodes of bound variables, and they depend on the alternation depth of the subformulas of $\phi$ which bind those variables. More precisely, consider a node $\pair{X}{x}$, where $X$ is bound in $\phi$ by $\eta X'.\Phi$ (i.e. it occurs on the left-hand side of one of the equations from $\Phi$), for $\eta\in\{\mu,\nu\}$. Then
\begin{itemize}
\item if $\eta = \mu$, then the rank of $\pair{X}{x}$ is $2\cdot\lfloor \frac{\alpha(\mu X'.\Phi)}{2}\rfloor+1$ and
\item if $\eta = \nu$, then the rank of $\pair{X}{x}$ is $2\cdot\lfloor \frac{\alpha(\nu X'.\Phi)}{2}\rfloor$.
\end{itemize}
Now we are ready to prove
\begin{thm}[Adequacy theorem]\label{thm:adequacy}
	For every state $x$ in a Kripke model $\K$ and every formula $\phi$ in $\vaform$,
    \[x\models^{\K} \phi \iff \pair{\phi}{x}\in\Win_{\exists}(\evalg{\phi}).\]
\end{thm}
\begin{proof}
We closely follow the lines of~\cite[Thm.~3.27]{yde}. That proof relies on an auxiliary notion of an \emph{unfolding game}, which needs to be  generalized a little to adapt it to the present context, where fixpoints are taken over families of sets indexed by atoms. This requires only a very small adjustment of the original conctruction.
Given a monotone operator
\[F: \power(A)^B\rightarrow \power(A)^B,\]
we can see it as a monotone operator $F'$ of type
\[\power(B\times A) \rightarrow \power(B\times A).\]
With the above definition one quickly checks that
\begin{align*}
\pair{b}{a}\in \text{lfp}(F') &\iff a\in (\text{lfp}(F))_b.\\
\pair{b}{a}\in \text{gfp}(F') &\iff a\in (\text{gfp}(F))_b.
\end{align*}
Now, the definition the unfolding game of $F$ is as in \cite{yde} (we sketch it for the reader's convenience):  it is a simple parity game with the set of nodes 
\[(B\times A) \cup \power(B\times A).\]
Nodes in $B\times A$ are controlled by $\exists$, the remaining nodes are controlled by $\forall$. From a node $\pair{b}{a}\in B\times A$, the player $\exists$ is allowed to move to any set $C\subseteq B\times A$ such that 
\begin{equation*}\label{equat:move_unfg}
\pair{b}{a}\in F(C).
\end{equation*}
Conversely, given a set $C\subseteq B\times A$, the player $\forall$ chooses any $\pair{b}{a}$ such that
\[\pair{b}{a}\in C.\]

The winning condition depends on whether we want to characterize the greatest or the least fixpoint of $F$. In the former case, every infinite game is won by $\exists$ (i.e. all positions are assigned rank $0$), in the latter, such games are won by $\forall$ (i.e. all positions are assigned rank~$1$). In both cases the player who got stuck, loses.

The unfolding game for the least fixpoint of $F$ will be denoted by $\mathcal{U}^{\mu}_{F}$; the game for the greatest fixpoint by $\mathcal{U}^{\nu}_{F}$.

The following proposition is proved by a mix of routine fixpoint arguments; see the proof of~\cite[Thm.~3.14]{yde}.


\begin{stw}\label{prop:unfg}
For a monotone $F:\power(A)^B\rightarrow \power(A)^B$, let $\text{lfp}(F)$ and $\text{gfp}(F)$ denote the least and the greatest fixpoints of $F$. For any $a\in A$ and $b\in B$,
\begin{itemize}
\item $\pair{a}{b}\in\Win_{\exists}(\mathcal{U}^{\mu}_{F})$ if and only if  $a\in (\text{lfp}(F))_b$, and
\item $\pair{a}{b}\in\Win_{\exists}(\mathcal{U}^{\nu}_{F})$ if and only if  $a\in (\text{gfp}(F))_b$.
\qed
\end{itemize}
\end{stw}

As in~\cite[Thm.~3.27]{yde}, a proof of Theorem~\ref{thm:adequacy} follows by induction on the structure of the formula $\phi$. At each step, allowing a nonempty but orbit-finite set $\was{X_j}_{j\in J}$ of free variables and assuming their interpretations 
$\was{A_j}_{j\in J}$ as subsets $A_j\subseteq K$, one proves that:
\begin{equation}\label{equat:ind_evalg}
x\in \sembr{\phi}_{\was{X_j\mapsto A_j}_{j\in J}} \iff \pair{\phi}{x}\in\Win_{\exists}(\mathcal{G}_{\phi,\mathcal{K}[X_j\mapsto A_j]_{j\in J}}).
\end{equation}

The only case that differs slightly from~\cite{yde} is the fixpoint operator: once again we have to adapt the proof to the case of vectorial calculus. Assume that $\phi = \eta X_i. \Phi$ where $\eta\in\{\mu,\nu\}$ and 
\[\Phi := \was{X_j.\psi_j}_{j\in J}\] 
is an orbit-finite set of equations with $i\in J$.
The monotone operator induced by $\Phi$ on $\power(K)^J$ will be denoted by $\sembr{\Phi}$. For every $j\in J$ we put 
\[\phi^j := \eta X_j.\Phi.\]
(Note that $\phi^i=\phi$.) By our inductive assumption, for every $j$ the condition \eqref{equat:ind_evalg} holds for $\psi_j$. By the definition of $\sembr{\eta X.\Phi}$ and by Proposition~\ref{prop:unfg} it is sufficient to demonstrate that for every $x\in K$ and $j\in J$:
\[
\pair{x}{j}\in \Win_{\exists}(\mathcal{U}^\eta_{\sembr\Phi}) \iff \pair{\phi^j}{x}\in\Win_{\exists}(\evalg{\phi^j}).
\]
The proof of the above fact (for the scalar case) given in ~\cite[Thm.~3.27]{yde} proceeds by translating winning strategies for $\exists$ between the two games. This proof adapts straightforwardly to our case.
\end{proof}
We should remark that in the above proof, the changes that we introduced were driven by the vectorial character of our formalism rather than directly by the use of atoms. One would have to make essentially the same changes in the original proof (i.e. \cite[Thm.~3.27]{yde}) to make it work for the standard (i.e. atom-less) vectorial $\mu$-calculus. As we already mentioned, in the atom-less setting this can be avoided since the atom-less vectorial calculus can be encoded into the scalar one. That encoding fails in the presence of atoms, as we shall show in Section \ref{sect_separation}.


\section{Atomic bisimulation games}\label{sect_bisimulations}
In Sections~\ref{sect_separation}-\ref{sect_limitations} we shall show that atomic $\mu$-calculi cannot define various properties of Kripke models. A general strategy for proving such results is using suitable notions of bisimulation: first prove that bisimilar states cannot be distinguished by any formula, then show a Kripke model and two bisimilar states $p$, $q$ in it such that $p$ has property $P$ but $q$ does not, to conclude that the property $P$ is not definable by any formula. 

In this section, to prepare the ground for such undefinability results, we define two notions of bisimulation (or, more precisely, two hierarchies of bisimulations parametrized by natural numbers). The first, which we call {\em stack bisimulation}, is suitable for scalar $\mu$-calculi with atoms. The second, simply called {\em bisimulation}, is suitable for the vector calculi.

Denote by $\Atoms^{(\leq k)}$ the set of ordered tuples of atoms of length at most $k$. Elements of such sets will be denoted with vector notation: $\vec{a},\vec{b}$ etc. If $|\vec{a}|>0$, let $\pop(\vec{a})$ denote the tuple $\vec{a}$ with the last element removed.
Recall that we write $x\sim y$ to say that $x$ and $y$ are in the same orbit.
\begin{defin}\label{def:k-stack-bisim}
For a number $k\in\mathbb{N}$, a {\em $k$-stack-bisimulation} on a Kripke model ${\mathcal{K}}$ is a symmetric relation $B$ on $K\times\Atoms^{(\leq k)}$ such that, whenever $\pair{x}{\vec{a}}B\pair{y}{\vec{b}}$ then:
\begin{enumerate}[label=(\roman*)]
\item $\pair{\pred(x)}{\vec{a}}\sim\pair{\pred(y)}{\vec{b}}$ (and in particular $|\vec{a}|=|\vec{b}|$),
\item for every $x'$ such that $x\step x'$ there is a $y'$ such that $y\step y'$ and $\pair{x'}{\vec{a}}B\pair{y'}{\vec{b}}$, 
\item if $|\vec{a}|>0$ then $\pair{x}{\pop(\vec{a})}B\pair{y}{\pop(\vec{b})}$,
\item if $|\vec{a}|<k$ then for every $c\in\Atoms$ there exists a $d\in\Atoms$ such that $\pair{x}{\vec{a}c}B\pair{y}{\vec{b}d}$.
\end{enumerate}
Two states are called {\em $k$-stack-bisimilar} if they are related by a $k$-stack-bisimulation.
\end{defin}
Formal definitions of particular bisimulations are often rather dry, and it is convenient to prove their existence using simple bisimulation games. A {\em $k$-stack-bisimulation game} on a model ${\mathcal{K}}$ is played between two players called Spoiler and Duplicator, who make their moves alternately. A legal position in the game is a tuple
\[
	\langle x,\vec{a},y,\vec{b} \rangle \in K\times\Atoms^{(\leq k)}\times K\times\Atoms^{(\leq k)}
\]
such that $\pair{\pred(x)}{\vec{a}}\sim\pair{\pred(y)}{\vec{b}}$ (and in particular $|\vec{a}|=|\vec{b}|$). In such a position, Spoiler can make three kinds of moves. He can either
\begin{itemize}
\item make a {\em model move}:
\begin{itemize}
\item choose a state $x'\in K$ such that $x\step x'$, to which Duplicator must respond by choosing a state $y'\in K$ such that $y\step y'$; or
\item choose a state $y'\in K$ such that $y\step y'$, to which Duplicator must respond by choosing a state $x'\in K$ such that $x\step x'$;
\end{itemize}
the game then proceeds from $\langle x',\vec{a},y',\vec{b} \rangle$ provided that it is a legal position; or
\item make a {\em pop move}, assuming that $|\vec{a}|>0$:
\begin{itemize}
\item remove the last element from $\vec{a}$, to which Duplicator must respond by removing the last element of $\vec{b}$; or
\item remove the last element from $\vec{b}$, to which Duplicator must respond by removing the last element of $\vec{a}$;
\end{itemize}
the game then proceeds from $\langle x,\pop(\vec{a}),y,\pop(\vec{b}) \rangle$, which is always a legal position; or
\item make a {\em push move}, assuming that $|\vec{a}|<k$:
\begin{itemize}
\item extend $\vec{a}$ with some atom $c$, to which Duplicator must respond by extending $\vec{b}$ with some atom $d$; or
\item extend $\vec{b}$ with some atom $d$, to which Duplicator must respond by extending $\vec{a}$ with some atom $c$;
\end{itemize}
the game then proceeds from $\langle x,\vec{a}c,y,\vec{b}d \rangle$, provided that it is a legal position.
\end{itemize}
Duplicator loses if she fails to move the game to a legal position in response to a Spoiler move. Since Spoiler can always make a pop move or a push move, the only way for Duplicator to win is to achieve an infinite play.

By a standard argument, $\pair{x}{\vec{a}}$ and $\pair{y}{\vec{b}}$ are $k$-stack-bisimilar if and only if Duplicator has a winning strategy in the $k$-stack-game from the position $\langle x,\vec{a},y,\vec{b}\rangle$.

We will now show that $k$-stack-bisimilar states cannot be distinguished by formulas from a certain fragment of the atomic $\mu$-calculus $\aform$. To describe this fragment, it is convenient to introduce a notion of an ``ordered'' global support of a subformula.

Given a fixed equivariant formula $\phi$, we define a relation between finite sequences of atoms and (occurrences of) subformulas of $\phi$. For $\vec{a}\in\Atoms^*$ and $\psi\subf\phi$, we write $\vec{a}\gsp_\phi\psi$ to read ``$\vec{a}$ globally supports $\psi$ within $\phi$''. This relation is defined by induction on the depth of the occurrence of $\psi$ within $\phi$:
\begin{itemize}
\item $\vec{a}\gsp_\phi\phi$ if and only if $\vec{a}=\epsilon$,
\item if $\psi\subf\phi$ such that $\psi=\diam\psi'$, $\psi=\Box\psi'$, $\psi=\mu X.\psi'$ or $\psi=\nu X.\psi'$ then $\vec{a}\gsp_\phi\psi'$ if and only if $\vec{a}\gsp_\phi\psi$,
\item if $\psi\subf\phi$ such that $\psi = \bigvee\Psi$ or $\psi=\bigwedge\Psi$ and $\psi'\in\Psi$, then $\vec{e}\gsp_\phi\psi'$ if and only if there is a (necessarily unique, for a given $\vec{e}$) decomposition $\vec{e}=\vec{a}\conc \vec{c}$ such that:
\begin{itemize}
\item $\vec{a}\gsp_\phi\psi$ and
\item $\vec{c}$ is some ordering of the set $\supp(\psi')$.
\end{itemize}
\end{itemize}
Some basic properties of global supports follow directly from the definition. Every (occurrence of) subformula $\psi\subf\phi$ is globally supported by some sequence, possibly more than one,
since the least supports of each formula on the path trom $\psi$ to the root of $\phi$ may be ordered in many ways. However, all global supports of any given $\psi$ within $\phi$ have the same length. Moreover, for subformulas $\psi\subf\theta\subf\phi$, 
\begin{itemize}
\item every global support of $\psi$ within $\phi$ has a (unique) prefix that globally supports $\theta$ within $\phi$, and
\item every global support of $\theta$ within $\phi$ extends (not necessarily uniquely) to a global support of $\psi$ within $\phi$.
\end{itemize}
It is also easy to see that if $\vec{a}\gsp_\phi\psi$ then the set of all atoms in $\vec{a}$ supports $\psi$, not only as a formula but also as its occurrence within $\phi$.

Call a formula $\phi$ {\em globally $k$-supported} if every subformula $\psi\subf\phi$ has a global support of length at most $k$. It is easy to see that every formula $\phi$ in $\aform$, or even in $\vaform$, is globally $k$-supported for some number $k$, bounded from above by the height of $\phi$ multiplied by the maximal size of the least support of a subformula of $\phi$.

The following theorem is formulated for equivariant models and formulas, but it is not difficult to generalize it to ones supported by some fixed set $S$.

\begin{thm}\label{thm:k-stack-bisim}
Assume that $\pair{x}{\epsilon}$ and $\pair{y}{\epsilon}$ are $k$-stack-bisimilar in an equivariant model $\K$. For every equivariant, globally $k$-supported formula $\phi$ in $\aform$, $x\models\phi$ if and only if $y\models\phi$.
\end{thm}
\begin{proof}
Fix $\phi$, $x$ and $y$ as in the assumptions. We will define a bisimulation relation $B$ (in the classical sense) on the graph of the game $\evalg{\phi}$ such that
\begin{equation}\label{eq:xBy}
	\pair{\phi}{x} B \pair{\phi}{y}.
\end{equation}
This theorem will then follow by Theorem~\ref{thm:adequacy}. Note that the graph of $\evalg{\phi}$ can be seen as a Kripke model where labels correspond to ranks of nodes together with the information about which player controls which node. Since $\K$ is fixed we will skip the reference to it. 
   \newcommand{\eval}[1]{\mathcal{G}_{#1}}
   
   The bisimulation $B$ is defined by: $\pair{\psi}{x'}B\pair{\theta}{y'}$ if and only if 
   \begin{enumerate}
   \item there exists $\pi\in\aut(\Atoms)$ such that (the occurrence) $\psi$ in $\phi$ is mapped by $\pi$ to (the occurrence) $\theta$ in $\phi$ (note that $\phi\cdot\pi=\phi$ since $\phi$ is equivariant), and
\item there exist $\vec{a},\vec{b}\in\Atoms^{(\leq k)}$ such that:
\begin{itemize}
\item $\vec{a}\gsp_\phi\psi$,
\item $\vec{b}\gsp_\phi\theta$,
\item $\pair{x'}{\vec{a}}$  and $\pair{y'}{\vec{b}}$ are $k$-stack bisimilar.
\end{itemize}	
    \end{enumerate}
Under this definition~\eqref{eq:xBy} indeed holds: it is enough to put the identity automorphism as $\pi$ and $\vec{a}=\vec{b}=\epsilon$.
    
We need to verify that $B$ is a bisimulation on $\eval{\phi}$. Note that if $\pair{\psi}{x'}B\pair{\theta}{y'}$ then both pairs get the same label (i.e.~rank) in $\eval{\phi}$, since 
the label of a formula depends only on its syntactical properties (i.e.~alternation depth), which are preserved by application of atom automorphims.
      
Now we verify the zig-zag bisimulation condition for $B$. Since $B$ is clearly symmetric, it is enough to verify one direction only. Assume that 
	\[
    \pair{\psi}{x'}B\pair{\theta}{y'}
    \qquad\text{and}\qquad
    \pair{\psi}{x'}\longrightarrow^{\eval{\phi}} \pair{\psi'}{x''}.
    \]
	 We are looking for $\pair{\theta'}{y''}$ such that 
	 \[
	 \pair{\psi'}{x''}B\pair{\theta'}{y''} \qquad\text{and}\qquad \pair{\theta}{y'}\longrightarrow^{\eval{\phi}} \pair{\theta'}{y''}.
	 \] 
	 We cover three non-trivial cases:
	\begin{enumerate}
		\item $\psi = \diam \xi$ (the case $\psi=\Box\xi$ is similar),
		\item $\psi = \bigvee\Psi$ (the case $\psi=\bigwedge\Psi$ is similar), or
		\item $\psi = X$ (for a bound $X$).
	\end{enumerate}
\paragraph{Case 1}
 In such a situation $\psi' = \xi$ and $x'\longrightarrow^{\K} x''$. Moreover, $\theta = \diam \theta'$ for some formula $\theta'$. Since $\pair{\psi}{x'}B\pair{\theta}{y'}$, there are some sequences $\vec{a},\vec{b}\in\Atoms^{(\leq k)}$ such that $\vec{a}\gsp_\phi\psi$, $\vec{b}\gsp_\phi\theta$ and $\pair{x'}{\vec{a}}$ and $\pair{y'}{\vec{b}}$ are $k$-stack bisimilar. By condition (ii) in Definition~\ref{def:k-stack-bisim}, there is some state $y''$ in $\K$ such that $y'\longrightarrow^{\K} y''$ (hence $\pair{\theta}{y'}\longrightarrow^{\eval{\phi}} \pair{\theta'}{y''}$ as required) and $\pair{x''}{\vec{a}}$ and $\pair{y''}{\vec{b}}$ are $k$-stack bisimilar. By definition of the global support relation we have $\vec{a}\gsp_\phi\psi'$, $\vec{b}\gsp_\phi\theta'$, hence $\pair{\psi'}{x''}B\pair{\theta'}{y''}$ (the automorphism $\pi$ that was good for the pair $\psi,\theta$ remains good for $\psi'$, $\theta'$). 
 \paragraph{Case 2}
 In such a situation $x'' = x'$ and $\psi'\in\Psi$.  Moreover, $\theta=\bigvee\Theta$ for some orbit-finite set $\Theta$ of formulas. Since $\pair{\psi}{x'}B\pair{\theta}{y'}$, there are some sequences $\vec{a},\vec{b}\in\Atoms^{(\leq k)}$ such that $\vec{a}\gsp_\phi\psi$, $\vec{b}\gsp_\phi\theta$ and $\pair{x'}{\vec{a}}$ and $\pair{y'}{\vec{b}}$ are $k$-stack bisimilar. Let $\vec{c}$ be any ordering of the finite set $\supp(\psi')$. Then $\vec{a}\conc \vec{c}\gsp_\phi\psi'$, and since by our assumption $\phi$ is globally $k$-supported, the length of $\vec{a}\conc \vec{c}$ is at most $k$.
 
 Since $\pair{x'}{\vec{a}}$ and $\pair{y'}{\vec{b}}$ are $k$-stack bisimilar, by repeated application of condition (iv) in Definition~\ref{def:k-stack-bisim}, there is a sequence $\vec{d}\in\Atoms^*$ such that 
 \[
 \pair{x'}{\vec{a}\conc \vec{c}}\textnormal{ and }\pair{y'}{\vec{b}\conc \vec{d}} \textnormal{ are $k$-stack bisimilar}.
 \]
 This means in particular that $\vec{a}\conc \vec{c}\sim \vec{b}\conc \vec{d}$.
This implies that there exists some $\pi'\in\aut(\Atoms)$ such that:
\[
	\psi\cdot\pi'=\theta \quad\text{(hence }\Psi\cdot\pi'=\Theta\text{)}, \qquad \vec{a}\cdot\pi'=\vec{b} \qquad\text{and}\qquad \vec{c}\cdot\pi'=\vec{d}.
\]
Put $\theta'=\psi'\cdot\pi'$ and $y''=y'$. Then
\[
	\theta'\in\Theta \qquad\text{and}\qquad \vec{b}\conc \vec{d}\gsp_\phi\theta'.
\]
 As a result, 
 	 \[
	 \pair{\psi'}{x''}B\pair{\theta'}{y''} \qquad\text{and}\qquad \pair{\theta}{y'}\longrightarrow^{\eval{\phi}} \pair{\theta'}{y''}
	 \] 
as required.
 
 \paragraph{Case 3}
 In this situation $x''=x'$ and $\psi\subf\psi'=\mu X.\zeta$ is the subformula of $\phi$ that binds $X$. Since $\pi$ maps $\psi$ to $\theta$ as an occurrence of a subformula, we get that 
 \[\theta=X\cdot\pi\subf\theta'=\mu (X\cdot\pi).(\zeta\cdot\pi) \subf \phi.\]
 Clearly, $\psi'\cdot\pi=\theta'$, hence $\theta'$ is the subformula of $\phi$ that binds $X\cdot\pi$ and $\pair{\theta}{y'}\longrightarrow^{\eval{\phi}} \pair{\theta'}{y'}$. 
 
Since $\pair{\psi}{x'}B\pair{\theta}{y'}$, there are some sequences $\vec{a},\vec{b}\in\Atoms^{(\leq k)}$ such that $\vec{a}\gsp_\phi\psi$, $\vec{b}\gsp_\phi\theta$ and $\pair{x'}{\vec{a}}$ and $\pair{y'}{\vec{b}}$ are $k$-stack bisimilar. Let $\vec{c}$ be the (unique) prefix of $\vec{a}$ such that $\vec{c}\gsp_\phi\psi'$, and let $\vec{d}$ be the prefix of $\vec{b}$ with the same length as $\vec{c}$. Then $\vec{d}\gsp_\phi\theta'$. Indeed, by equivariance we have $(\vec{c}\cdot\pi)\gsp_\phi\theta'$, so all sequences that globally support $\theta'$ within $\phi$ must have the same length as $\vec{c}$; moreover, since $\vec{b}\gsp_\phi\theta$ we know that some prefix of $\vec{b}$ must globally support $\theta'$, and $\vec{d}$ is the prefix of $\vec{b}$ of the same length as $\vec{c}$.
 
By repeated use of condition (iii) in Definition~\ref{def:k-stack-bisim}, putting $y''=y'$, we get that $\pair{x''}{\vec{c}}$ and $\pair{y''}{\vec{d}}$ are $k$-stack bisimilar. From this we obtain
  	 \[
	 \pair{\psi'}{x''}B\pair{\theta'}{y''} \qquad\text{and}\qquad \pair{\theta}{y'}\longrightarrow^{\eval{\phi}} \pair{\theta'}{y''}
	 \] 
as required.
\end{proof}

Note that every formula in $\aform$ is globally $k$-supported for some number $k$. Therefore, by Theorem~\ref{thm:k-stack-bisim}, to prove that a property $P$ is not definable it is enough to construct, for every number $k$, a Kripke model ${\mathcal{K}}$ over $\Atoms$ as the language of basic predicate symbols, and two states $x,y\in K$ such that $\pair{x}{\epsilon}$ and $\pair{y}{\epsilon}$ are $k$-stack-bisimilar and $P$ fails for $x$ but holds for $y$. 

As we shall see in Section~\ref{sect_separation}, formulas of vectorial atomic $\mu$-calculus $\vaform$ can distinguish between stack-bisimilar states in a Kripke model. Therefore, to show undefinability results for vector calculi, we shall now introduce a finer notion of bisimulation, where in the corresponding bisimulation games Spoiler can replace atoms in the current position at will, not respecting the stack regime of Definition~\ref{def:k-stack-bisim}. 

\begin{defin}\label{def:k-bisim}
For a number $k\in\mathbb{N}$, a {\em $k$-bisimulation} on a Kripke model ${\mathcal{K}}$ is a symmetric relation $B$ on $K\times\Atoms^{(\leq k)}$ such that, whenever $\pair{x}{\vec{a}}B\pair{y}{\vec{b}}$ then:
\begin{enumerate}[label=(\roman*)]
\item $\pair{\pred(x)}{\vec{a}}\sim\pair{\pred(y)}{\vec{b}}$,
\item for every $x'$ such that $x\step x'$ there is a $y'$ such that $y\step y'$ and $\pair{x'}{\vec{a}}B\pair{y'}{\vec{b}}$, and
\item for every $\vec{c}\in\Atoms^{(\leq k)}$ there exists a $\vec{d}\in\Atoms^{(\leq k)}$ such that $\pair{\vec{a}}{\vec{c}}\sim\pair{\vec{b}}{\vec{d}}$ and
$\pair{x}{\vec{c}}B\pair{y}{\vec{d}}$.
\end{enumerate}
Two states are called {\em $k$-bisimilar} if they are related by a $k$-bisimulation.
\end{defin}

As for stack bisimulations above, $k$-bisimilarity on a Kripke model is an equivalence relation. If shorter tuples $\vec{a}'$ and $\vec{b}'$ arise from $\vec{a}$ and $\vec{b}$ respectively by selecting the same subset of positions, then $\pair{x}{\vec{a}'}$ and $\pair{y}{\vec{b}'}$ are $k$-bisimilar as well. Finally, for $l<k$, the restriction of a $k$-bisimulation to the set $K\times\Atoms^{(\leq l)}$ is an $l$-bisimulation. 

As before, a {\em $k$-bisimulation game} on a model ${\mathcal{K}}$ is played between Spoiler and Duplicator, and a legal position in the game is a tuple
\[
	\langle x,\vec{a},y,\vec{b} \rangle \in K\times\Atoms^{(\leq k)}\times K\times\Atoms^{(\leq k)}
\]
such that $\pair{\pred(x)}{\vec{a}}\sim\pair{\pred(y)}{\vec{b}}$. In such a position, Spoiler can make two kinds of moves. He can either
\begin{itemize}
\item make a {\em model move}:
\begin{itemize}
\item choose a state $x'\in K$ such that $x\step x'$, to which Duplicator must respond by choosing a state $y'\in K$ such that $y\step y'$; or
\item choose a state $y'\in K$ such that $y\step y'$, to which Duplicator must respond by choosing a state $x'\in K$ such that $x\step x'$;
\end{itemize}
the game then proceeds from $\langle x',\vec{a},y',\vec{b} \rangle$ provided that it is a legal position; or
\item make an {\em atom replacement move}:
\begin{itemize}
\item choose any tuple $\vec{c}\in\Atoms^{(\leq k)}$, to which Duplicator must respond by choosing a tuple $\vec{d}\in\Atoms^{(\leq k)}$ such that $\pair{\vec{a}}{\vec{c}}\sim\pair{\vec{b}}{\vec{d}}$; or
\item choose any tuple $\vec{d}\in\Atoms^{(\leq k)}$, to which Duplicator must respond by choosing a tuple $\vec{c}\in\Atoms^{(\leq k)}$ such that $\pair{\vec{a}}{\vec{c}}\sim\pair{\vec{b}}{\vec{d}}$;
\end{itemize}
the game then proceeds from $\langle x,\vec{c},y,\vec{d} \rangle$ provided that it is a legal position.
\end{itemize}
Duplicator loses if she fails to move the game to a legal position in response to a Spoiler move. Since Spoiler can always make an atom replacement move, the only way for Duplicator to win is to achieve an infinite play.

By a standard argument, $\pair{x}{\vec{a}}$ and $\pair{y}{\vec{b}}$ are $k$-bisimilar if and only if Duplicator has a winning strategy from the position $\langle x,\vec{a},y,\vec{b}\rangle$.

The following result shows that $k$-bisimilar states cannot be distinguished by formulas from the globally $k$-supported fragment of the vectorial atomic $\mu$-calculus $\vaform$. As for Theorem~\ref{thm:k-stack-bisim}, is not difficult to generalize it to models and formulas supported by some fixed set $S$.

\begin{thm}\label{thm:k-bisim}
Assume that $\pair{x}{\epsilon}$ and $\pair{y}{\epsilon}$ are $k$-bisimilar in an equivariant model $\K$. For every equivariant, globally $k$-supported formula $\phi$, $x\models\phi$ if and only if $y\models\phi$.
\end{thm}
\begin{proof}
The idea is the same as in the proof of Theorem~\ref{thm:k-stack-bisim}: for any globally $k$-supported formula $\phi$ in $\vaform$, we define a bisimulation $B$ on the graph of the game
$\evalg{\phi}$ such that $\pair{\phi}{x}$ and $\pair{\phi}{y}$ are related by $b$. 

The relation $B$ is defined exactly as before, with $k$-stack bisimulation replaced by $k$-bisimulation. In checking the bisimulation condition, the only difference is the case of bound variables (i.e.~Case 3), arising from the fact that the binding fixpoint constructions are now richer than those in the scalar calculus. Assume that 
 \[\pair{X}{x'}B\pair{\theta}{y'}\]
 and $X$ is a variable bound in $\mu Y.\Phi\subf\phi$. By definition of $B$, there is some $\pi\in\aut(\Atoms)$ such that $X\cdot\pi=\theta$, so $\theta$ is a variable bound in some $\mu (Y\cdot\pi).(\Phi\cdot \pi)\subf\phi$. 
 
The game $\evalg{\phi}$ makes an automatic move from the node $\pair{X}{x'}$ to the node $\pair{\psi'}{x'}$, where $\psi'$ is the unique formula such that $X.\psi'\in\Phi$. Note, however, that $X$ is not necessarily a subformula of $\psi'$, so there is no guarantee that the automatic move from $\pair{\theta}{y'}$ goes to the node $\pair{\psi'\cdot\pi}{y'}$.

Since $\pair{X}{x'}B\pair{\theta}{y'}$, there are some sequences $\vec{a},\vec{b}\in\Atoms^{(\leq k)}$ such that $\vec{a}\gsp_\phi X$, $\vec{b}\gsp_\phi\theta$ and $\pair{x'}{\vec{a}}$ and $\pair{y'}{\vec{b}}$ are $k$-bisimilar. Moreover, let $\vec{c}\in A^{(\leq k)}$ be such that $\vec{c}\gsp_\phi\psi'$. Note that it may not be possible to choose $\vec{c}$ which is a prefix of $\vec{a}$; this is the crucial difference from the proof of Theorem~\ref{thm:k-stack-bisim}.

Since $\pair{x'}{\vec{a}}$ and $\pair{y'}{\vec{b}}$ are $k$-bisimilar, thanks to condition (iii) in Definition~\ref{def:k-bisim} we can choose $\vec{d}$ such that $\pair{\vec{a}}{\vec{c}}\sim\pair{\vec{b}}{\vec{d}}$ and such that $\pair{x'}{\vec{c}}$ and $\pair{y'}{\vec{d}}$ are $k$-bisimilar. So pick some $\tau\in\aut(\Atoms)$ such that $\vec{a}\cdot\tau=\vec{b}$ and $\vec{c}\cdot\tau=\vec{d}$. By definition of $\vec{a}$ we have that
\[
	X\cdot\tau = X\cdot\pi = \theta
    \qquad\text{and} \quad
    \Phi\cdot\tau = \Phi\cdot\pi.
\]
As a result, $\psi'\cdot\tau$ is the 
formula such that the equation 
\[\theta.(\psi'\cdot\tau)\]
is in $\Phi\cdot\tau$. Also, since $\vec{c}\gsp_\phi\psi'$, hence $\vec{d}\gsp_\phi(\psi'\cdot\tau)$. Hence putting $\theta' = \psi'\cdot\tau$ and $y'' = y'$ we are done.
\end{proof}

As we already observed, every formula of the vectorial atomic $\mu$-calculus $\vaform$ is globally $k$-supported for some number $k$. Therefore, by Theorem~\ref{thm:k-bisim}, to prove that a property $P$ is not definable it is enough to construct, for every number $k$, an orbit-finite Kripke model ${\mathcal{K}}$ over $\Atoms$ as the language of basic predicate symbols, and two states $x,y\in K$ such that $\pair{x}{\epsilon}$ and $\pair{y}{\epsilon}$ are $k$-bisimilar and $P$ fails for $x$ but holds for $y$.


\section{Undecidability results}\label{sect_undecidability}
The main purpose of this section is to show that the satisfiability of atomic $\mu$-calculus formulas is undecidable, even for the scalar calculus $\aform$. We start by showing that satisfiability of atomic LTL formulas is undecidable and end with the same conclusion about the model checking problem for CTL$^{*}$. All proofs in this section work the same both for equality and ordered atoms.

\begin{thm}\label{thm:ltl-sat-undec}
It is undecidable whether a given atomic LTL formula is satisfiable.
\end{thm}
\begin{proof}
The proof follows the lines of the proof from~\cite{NSV01} of the undecidability of the universality problem for register automata. 

Given a Turing machine $\mathcal{M}$ with a (finite) set of states $Q$, a (finite) working alphabet $\Gamma$ (including a blank symbol $\flat$), initial and accepting states $q_{init},q_{acc}\in Q$ and a transition relation $\delta$, we shall build an atomic LTL formula that is satisfiable if and only if $\mathcal{M}$ accepts the empty word. Without loss of generality assume that $\mathcal{M}$, upon reaching an accepting configuration, enters an infinite loop in that configuration.

Consider the following language of basic predicate symbols:
\begin{itemize}
\item a special symbol $\$$,
\item ${\tt atom}_a$ for each $a\in\Atoms$,
\item ${\tt tape}_\gamma$ for each $\gamma\in\Gamma$,
\item ${\tt head}_q$ for each $q\in Q$.
\end{itemize}
Our formula shall be a conjunction of several properties. First, we enforce that every state in a model either satisfies $\$$ (and no other basic predicates) or satisfies exactly one predicate ${\tt atom}_a$, exactly one predicate ${\tt tape}_{\gamma}$ and at most one predicate ${\tt head}_q$. This is ensured by a conjunction of formulas such as:
\[
{\sf G}\left(\$\vee \bigvee_{a\in\Atoms}{\tt atom}_a\right) \quad\text{and}\quad\bigwedge_{a\neq b\in\Atoms}{{\sf G}({\tt atom}_a\to \neg {\tt atom}_b)}
\]
and so on. We also ensure that the initial state of the model satisfies $\$$.

Furthermore, we ensure that as far as predicates $\$$ and ${\tt atom}_a$ are concerned, the model can be presented as an infinite word:
\[
	\$w\$w\$w\cdots
\]
where $w$ is a finite sequence of predicate symbols of the form ${\tt atom_a}$ such that no single predicate appears in $w$ more than once. This is achieved by a conjunction of formulas such as:
\begin{itemize}
\item $\bigwedge_{a,b\in\Atoms}{\sf G}({\tt atom}_a\land {\sf X}{\tt atom}_b \to {\sf G}({\tt atom}_a\to {\sf X}{\tt atom}_b))$
\item $\bigwedge_{a\in\Atoms}{\sf G}({\tt atom}_a\to (\neg {\tt atom}_a {\sf U} \$))$.
\end{itemize}

Portions of the model between two consecutive occurrences of $\$$ will store configurations of $\mathcal{M}$. To this end, we ensure that each portion contains exactly one state that satisfies some ${\tt head}$ predicate:
\[
{\sf G}(\$\to{\sf X}((\neg \psi\land\neg\$) {\sf U} (\psi \land {\sf X}(\neg \psi {\sf U} \$))))
\]
where $\psi = \bigvee_{q\in Q}{\tt head}_q$.

We then ensure that every two consecutive configurations encode a legal step of $\mathcal{M}$. Predicates ${\tt atom}$ are useful for this, as they trace single tape cells in subsequent configurations. To ensure that the letters on the tape do not change unless the machine head is directly over them, we state for each $\gamma\in\Gamma$:
\[
	\bigwedge_{a\in\Atoms} {\sf G}({\tt atom}_a\land{\tt tape}_{\gamma}\land\neg\psi\to {\sf X}(\neg{\tt atom}_a {\sf U} ({\tt atom}_a\land {\tt tape}_{\gamma})))
\]
where $\psi$ is as before.

Furthermore, for every ``head to the right'' rule
\[
	\langle q,\gamma,q',\gamma',\Rightarrow\rangle \in \delta
\]
we add a formula
\[
	\bigwedge_{a\in\Atoms}{\sf G}({\tt atom}_a\land {\tt head}_q\land{\tt tape}_{\gamma}\to {\sf X}(\neg {\tt atom}_a {\sf U} ({\tt atom}_a\land {\tt tape}_{\gamma'} \land {\sf X}{\tt head}_{q'})))
\]
and similarly for ``head to the left'' transition rules.

Finally we ensure the correct form of the initial configuration and that an accepting state is reached, by adding formulas ${\sf X}{\tt head_{q_{init}}}$, ${\sf X}({\tt tape}_{\flat}{\sf U}\$)$ and ${\sf F}{\tt head}_{q_{acc}}$.

Models of the conjunction of all these formulas correspond to accepting runs of $\mathcal{M}$ on the empty input word. As a result, it is undecidable whether an LTL formula has a model.
\end{proof}

Properties used in the above proof can be formulated in atomic $\mu$-calculus as well. 
As an immediate result, it is undecidable whether a given atomic $\mu$-calculus formula has a {\em deterministic} model. Some more care is needed to prove that general satisfiability is undecidable, but:

\begin{thm}\label{thm:mu-sat-undec}
It is undecidable whether a given formula of the atomic $\mu$-calculus is satisfiable.
\end{thm}
\begin{proof}
For any LTL formula $\phi$ in the negation normal form (i.e., one where negations occur only in front of basic predicates, and where conjunction, disjunction, ${\sf X}$, ${\sf U}$ and ${\sf R}$ modal operators are used), construct an atomic $\mu$-calculus formula $M(\phi)$ by induction as follows:
\begin{align*}
 M(\top) &= \top, \qquad M(\bot)=\bot \\
 M(p) &= p, \qquad M(\neg p) = \neg p \\
 M({\sf X}\phi) &=  \Box M(\phi) \\
 M(\phi\lor\psi) &= M(\phi)\lor M(\psi) \\
 M(\phi\land\psi) &= M(\phi)\land M(\psi) \\
 M(\phi {\sf U}\psi) &= \mu Y.(M(\psi)\lor(M(\phi)\land\Box Y)) \\
 M(\phi {\sf R}\psi) &= \nu Y.(M(\psi)\land(M(\phi)\lor\Box Y))
\end{align*}
This translation does not have all properties that may be desired (see e.g.~\cite{CGR11}) but it is sufficient for our purposes:
\begin{enumerate}[label=(\roman*)]
\item In every word model, if a state satisfies $\phi$ then it satisfies $M(\phi)$, 
\item In every Kripke model $\K$, if a state $x$ satisfies $M(\phi)$ then every infinite path in $\K$ that starts from $x$, considered as a word model, satisfies $\phi$.
\end{enumerate}
Note that the converse to the implication in (ii) does not hold in general (consider e.g. $\phi={\sf X}p \lor {\sf X}q$ and its translation $M(\phi)=\Box p\lor \Box q$). Both properties (i) and (ii) are proved by induction, for example for (ii):
\begin{itemize}
\item Assume $x\models M(\phi\lor\psi)$. Without loss of generality, assume $x\models M(\phi)$. By the inductive assumption, every path starting from $x$ satisfies $\phi$, so it satisfies $\phi\lor\psi$.
\item Assume $x\models M(\phi {\sf U}\psi)$. By definition of $M$, on every path $\pi$ starting from $x$ the formula $M(\phi)$ holds in every state $y$ until at some point $M(\psi)$ holds. Now, for each state $y$ the path $\pi$ has a sub-path that starts at $y$, and by the inductive assumption $\phi$ holds for all these subpaths until at some point $\psi$ holds. As a result, $\phi {\sf U}\psi$ holds for the path $\pi$.
\end{itemize}

Properties (i) and (ii) immediately imply that $\phi$ is satisfiable if and only if $M(\phi)$ is satisfiable in a state where some infinite path begins. Putting $\psi = \nu X.\diam X$, a formula which holds exactly in those states where an infinite path begins, we get that $\phi$ is satisfiable if and only if $M(\phi) \land \psi$ is satisfiable. By Theorem~\ref{thm:ltl-sat-undec}, satisfiability of atomic $\mu$-calculus formulas is undecidable.
\end{proof}

Finally, as a corollary of Theorem~\ref{thm:ltl-sat-undec} we obtain:

\begin{thm}\label{thm:ctl-mc-undec}
The model checking problem for atomic CTL$^*$ is undecidable.
\end{thm}
\begin{proof}
We reduce the satisfiability problem for atomic LTL, with some insight into the proof of Theorem~\ref{thm:ltl-sat-undec}. Given any Turing machine $\mathcal{M}$ as considered there, consider a Kripke model with the set of states:
\[
	(\Atoms\times\Gamma\times(Q\cup\{\tt{nohead}\}))\cup\{\$\}
\]
with basic predicates defined in the obvious way, and with transitions going both ways between every two states. This model is orbit-finite and equivariant.

Now, for the atomic LTL formula $\phi$ obtained from $\mathcal{M}$ as in the proof of Theorem~\ref{thm:ltl-sat-undec}, the formula $\exists\phi$ holds in this model in the state $\$$ if and only if $\phi$ is satisfiable. 

This works regardless of whether we interpret CTL* path formulas over arbitrary paths or over finitely supported ones, because the formula $\phi$ in the proof of Theorem~\ref{thm:ltl-sat-undec} forces its model to be finitely supported. 
\end{proof}

As a corollary we conclude that, unlike in the classical, atom-less setting, atomic CTL$^*$ is {\em not} a fragment of the atomic $\mu$-calculus.


\section{Separation results}\label{sect_separation}
In previous sections we introduced four atomic $\mu$-calculi altogether: the scalar $\eqaform$ and the vectorial $\veqaform$ calculus for equality atoms, and their corresponding versions $\ordaform$ and $\vordaform$ for ordered atoms. It is clear that the vectorial versions are at least as expressive as the scalar ones, and that ordered-atoms calculi can express all properties that equality-atoms calculi can. In this section we show that all four calculi are in fact different, and we show three properties of Kripke models, called \infsucc{}, \chain{} and \evensucc{}, that separate the four calculi as shown:
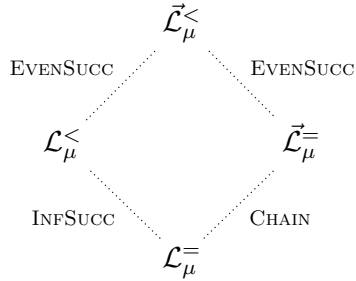
\[\xymatrix{
& \vordaform \\
\ordaform\ar@{.}[ur]^{\evensucc} & & \veqaform\ar@{.}[ul]_{\evensucc} \\
& \eqaform\ar@{.}[ul]^{\infsucc}\ar@{.}[ur]_{\chain}
}\]
Later, in Section~\ref{sect_limitations}, we shall show a further property that is not definable even in $\vordaform$.

\subsection{The \infsucc{} property}\label{sec:infsucc}{\ }

It is obvious that atomic $\mu$-calculus over ordered atoms is able to express many properties that are not definable over equality atoms simply because they refer to the ordering of atoms. However, there are properties that do not rely on that order {\em per se}, but are still definable in $\ordaform$ and not even in the vectorial calculus $\veqaform$.

Consider a vocabulary of basic predicates that coincides with the set $\Atoms$ of atoms. Let \infsucc{} denote the property ``there are infinitely many atoms $a$ such that some immediate successor of the current state satisfies $a$''.

Atom ordering is not mentioned anywhere in the definition on $\infsucc$, and indeed it is an equivariant property of (states in) Kripke models over equality atoms. However, it is not definable in $\veqaform$. To see this, we rely on Theorem~\ref{thm:k-bisim}. 

For any number $k$, we construct a model $\K$ over equality atoms and two states $p,q\in K$ that are $k$-bisimilar in $k$, but such that \infsucc{} holds for $p$ and fails for $q$. To this end, fix some finite set $S\subseteq\Atoms$ with $|S|=2k$. The model will be supported by $S$ (see below how to achieve a similar result with an equivariant model). Let the set of states of $\K$ be:
\[
	K = \{p,q\} + \{r_a\mid a\in\Atoms\}
\]
with the obvious action of $\aut(\Atoms)$ that makes each of $p$ and $q$ a singleton orbit. The satisfaction relation of $\K$ is such that
\[
	\pred(p)=\pred(q)=\emptyset \qquad \text{and} \qquad \pred(r_a)=\{a\} \quad\text{for } a\in\Atoms,
\]
and the transition relation is defined by
\[
p\step r_a \text{ for }  a\not\in S \qquad \text{and} \qquad q\step r_a \text{ for } a\in S.
\]
It is clear that the property \infsucc{} holds for $p$ but fails for $q$. However, $\pair{p}{\epsilon}$ and $\pair{q}{\epsilon}$ are $k$-bisimilar according to Definition~\ref{def:k-bisim}. To see this we describe a winning strategy for Duplicator in the corresponding $k$-bisimulation game.

Starting from the initial position $\langle p,\epsilon,q,\epsilon\rangle$, Spoiler may begin by making a series of atom replacement moves. Duplicator responds to each of these, ensuring that every position reached:
\[
	\langle p,\vec{a},q,\vec{b}\rangle
\]	
is legal, i.e. that $\vec{a}\sim\vec{b}$ (hence in particular $|\vec{a}|=|\vec{b}|$), and that the condition
\[
	a_i\not\in S \iff b_i\in S
\]
holds for each $i=1,\ldots,|\vec{a}|$. It is easy to see that this is possible thanks to the assumption that $|\vec{a}|\leq k$ and $|S|=2k$.

Afterwords, when Spoiler decides to make a model move from (say) $p$ to some $r_a$ such that $a=a_i$ for some $i\leq|\vec{a}|$, Duplicator responds by moving from $q$ to $b_i\in\vec{b}$. Thanks to Duplicator's previous effort, this is always a legal move. Other model moves from Spoiler are dealt with similarly. Afterwords, Spoiler is reduced to atom replacement moves that Duplicator can parry easily.

Our model $\K$ is supported by $S$, but it is easy to achieve a similar results in an equivariant model that includes a copy of our $\K$ for each $S\subseteq\Atoms$ with $|S|=2n$, with two additional equivariant initial states that can make transitions to the states $p$ and $q$, respectively, in each of the copies.

As a result, by Theorem~\ref{thm:k-bisim}, \infsucc{} is not definable in $\veqaform$.

Over ordered atoms, however, the situation is quite different. Notice that in a Kripke model with atoms, the set of successors of any particular state (and therefore the set of basic predicates that hold in those successors) is finitely supported. Recall that over ordered atoms, finitely supported subsets of $\Atoms$ are exactly those that are finite unions of open intervals and single points. As a result, a finitely supported set of atoms is infinite if and only if it contains an open interval. Using this observation, it is easy to see that \infsucc{} is definable in $\ordaform$ by a formula:
\[
	\bigvee_{a<b\in\Atoms}\bigwedge_{c\in(a;b)}\diam c.
\]

\subsection{The \chain{} property}{\ }

As before, consider a vocabulary of basic propositions that coincides with the set of $\Atoms$ of atoms (in this section we consider equality atoms only). Given a Kripke model $\K$, an infinite path:
\[
	x_0 \step x_1 \step x_2 \step \cdots
\]
in $\K$ is called a {\em chain} if 
\[
	|\pred(x_i)|=1 \qquad \text{and} \qquad \pred(x_{2i}) = \pred(x_{2i+3})
\]
for each $i\in\mathbb{N}$. The following diagram, where required equality of basic propositions satisfied in states is marked with dashed lines, justifies the name {\em chain}:

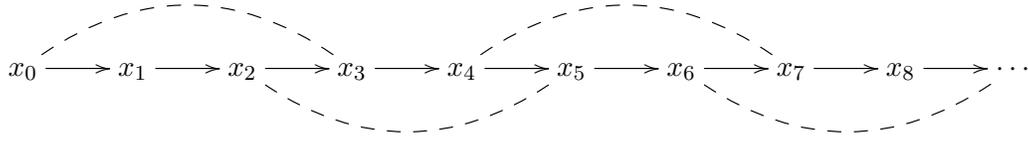
\hspace{2ex}

\[
\xymatrix{
x_0\ar[r]\ar@/^2pc/@{--}[rrr] & x_1\ar[r] & x_2\ar[r]\ar@/_2pc/@{--}[rrr] & x_3\ar[r] & x_4\ar[r]\ar@/^2pc/@{--}[rrr] & x_5\ar[r] & x_6\ar[r]\ar@/_2pc/@{--}[rrr]  & x_7\ar[r] & x_8\ar[r] & \cdots \\
{ }
}
\]

Let \chain{} denote the property ``there exists a chain path starting from the current state''. To show that it is not definable in $\eqaform$, we rely on Theorem~\ref{thm:k-stack-bisim}. We  construct, for every number $k$, an orbit-finite Kripke model ${\mathcal{K}}$ over $\Atoms$ as the language of basic predicate symbols, and two states $p,q\in K$ such that $\pair{p}{\epsilon}$ and $\pair{q}{\epsilon}$ are $k$-stack-bisimilar and \chain{} holds for $p$ but fails for $q$. 

Given a positive number $k$, pick some large $n$ ($n> 2^k$ will be enough) and fix $2n$ distinct atoms 
\[
	S = \{a_0,a_1,\ldots,a_{n-1},a_n\}\cup\{b_1,\ldots,b_{n-1}\}.
\] 
We shall build a (finite) model $\K$ supported by $S$. Its set $K$ of states will contain:
\begin{itemize}
\item $p_i$ and $q_i$ for $i=1,\ldots,2n$,
\item $r_i$ and $s_i$ for $i=1,\ldots,n-1$,
\item two special states $\top$ and $\bot$.
\end{itemize}
The satisfaction relation of $\K$ is defined so that:
\begin{itemize}
\item $\pred(p_i)=\pred(q_i)=\{a_{\frac{i}{2}-1}\}$ if $i$ is even,
\item $\pred(p_i)=\pred(q_i)=\{a_{\frac{i+1}{2}}\}$ if $i$ is odd,
\item $\pred(r_i)=\pred(s_i)=\{b_i\}$,
\item $\pred(\top)=a_n$ and $\pred(\bot)=\emptyset$.
\end{itemize}
The transition relation of $\K$ contains the following:
\begin{itemize}
\item $p_i\step p_{i+1}$ and $q_i\step q_{i+1}$ for $1\leq i<2n$,
\item $p_{2n}\step\top$ and $q_{2n}\step\bot$,
\item $p_{2i+1}\step r_i$ and $q_{2i+1}\step s_i$ for $1\leq i< n$,
\item $r_i\step q_{2i+3}$ and $s_i\step p_{2i+3}$ for $1\leq i<n-1$,
\item $r_{n-1}\step\bot$ and $s_{n-1}\step\top$,
\item $\top\step\top$.
\end{itemize}
These definitions look rather depressing, but the following diagram for $n=5$ should make the idea clear. Here, atomic predicates satisfied in a state are marked over its outgoing transitions:
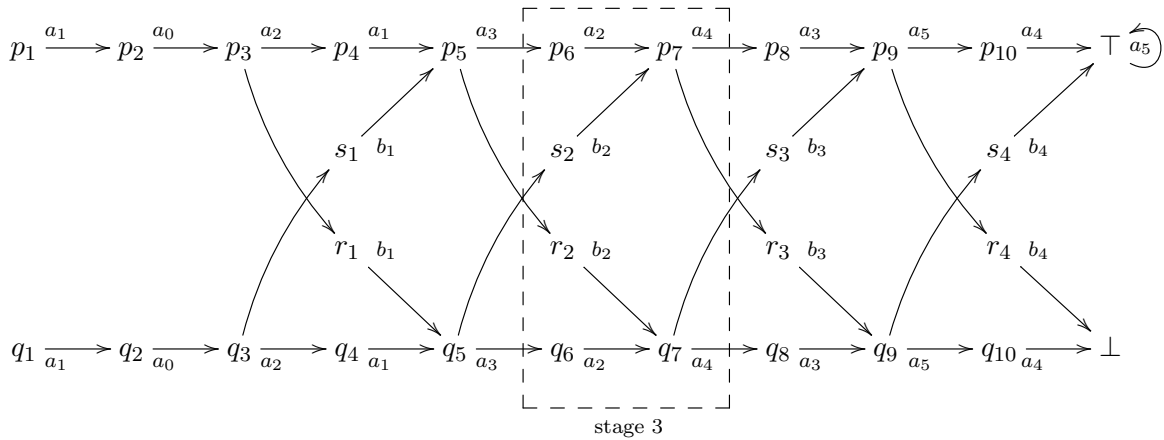
\[\xymatrix{
	p_1\ar[r]^(.3){a_1} & p_2\ar[r]^(.3){a_0} & p_3\ar[r]^(.3){a_2}\ar@/_/[ddr] & p_4\ar[r]^(.3){a_1} & p_5\ar[r]^(.3){a_3}\ar@/_/[ddr] & p_6\ar[r]^(.3){a_2} & p_7\ar[r]^(.3){a_4}\ar@/_/[ddr] & p_8\ar[r]^(.3){a_3} & p_9\ar[r]^(.3){a_5}\ar@/_/[ddr] & p_{10}\ar[r]^(.3){a_4} & \top\ar@(dr,ur)[]^{a_5} \\
	& & & s_1\ar[ur]_(.2){b_1} & & s_2\ar[ur]_(.2){b_2} & & s_3\ar[ur]_(.2){b_3} & & s_4\ar[ur]_(.2){b_4} \\
	& & & r_1\ar[dr]^(.2){b_1} & & r_2\ar[dr]^(.2){b_2} & & r_3\ar[dr]^(.2){b_3} & & r_4\ar[dr]^(.2){b_4} \\
	q_1\ar[r]_(.3){a_1} & q_2\ar[r]_(.3){a_0} & q_3\ar[r]_(.3){a_2}\ar@/^/[uur] & q_4\ar[r]_(.3){a_1} & q_5\ar[r]_(.3){a_3}\ar@/^/[uur] & q_6\ar[r]_(.3){a_2} & q_7\ar[r]_(.3){a_4}\ar@/^/[uur] & q_8\ar[r]_(.3){a_3} & q_9\ar[r]_(.3){a_5}\ar@/^/[uur] & q_{10}\ar[r]_(.3){a_4} & \bot \\
\save "1,6"+<-15pt,15pt>."4,7"*+<.5cm>[F--]\frm{} \restore
\save "4,6"+<25pt,-30pt>*\txt{\scriptsize{stage $3$}}\restore
}\]
(see below for an explanation of what {\em stage $3$} means).

It is important to note that the top-most path starting in $p_1$ and eventually entering an infinite loop at state $\top$, is a chain path. The state $p_1$ therefore satisfies the property \chain. On the other hand, the state $q_1$ does not satisfy the property: to build an infinite path from $q_1$ one has to visit a state $s_i$ at some point, and when that happens the chain property is immediately invalidated, since the basic predicate $b_i$ satisfied in $s_i$ does not match the predicate $a_i$ that held three steps earlier. 

We shall show that $p_1$ and $q_1$ are $k$-stack-bisimilar. We begin by showing a strategy for Duplicator in the $k$-stack-bisimulation game that, although not successful, will be the basis for constructing a better, winning strategy.

Starting from the initial position $\langle p_1,\epsilon,q_1,\epsilon\rangle$, Duplicator may try to copy Spoiler's moves verbatim, that is:
\begin{itemize}
\item respond to a model move $p_i\step p_{i+1}$ with $q_i\step q_{i+1}$, to a model move $p_{2i+1}\step r_i$ with $q_{2i+1}\step s_i$, etc.,
\item respond to a push move with an atom $c$ with a push move with the same atom $c$. 
\end{itemize}
(Spoiler's pop moves do not leave any choice for Duplicator, so there is no need to specify Duplicator's response to them.)

This strategy is losing: Spoiler can make a sequence of model moves from $p_1$ to $p_2,\ldots,p_{2n}$ and finally to $\top$ and then loop there, to which Duplicator responds by moving from $q_1$ all the way to $\bot$ and then she is stuck.

However, along the way, Duplicator may look for an opportunity to ``cheat''. To describe such opportunities, let us divide states in our model into {\em stages}, where stage $i$ (for $i=1,\ldots,n-1$) consists of states:
\[
	p_{2i}, q_{2i}, p_{2i+1}, q_{2i+1}, r_{i-1} \text{ and } s_{i-1}
\]
(the latter two are included only if $i>1$).
For example, stage $3$ is marked in the diagram above. Moreover, we say that an atom is {\em remembered} in a position if it is present in a sequence of atoms $\vec{a}$ that is a component of that position. Note that, according to Duplicator's strategy so far, the two sequences of atoms in every reachable position are the same, so at most $k$ atoms are remembered in any given position.

Since model moves in the game must be matched by a model moves, the two states in a game position always belong to the same stage; we can therefore meaningfully say that a position belongs to a stage. An opportunity for Duplicator to cheat occurs when the play reaches a position in a stage $i$, such that neither the atom $a_i$ nor $b_i$ is remembered.

If such a position is reached, Duplicator starts to respond to Spoiler's moves in a way that arises from applying the swap permutation $(a_i\ b_i)$, leaving all other atoms intact. More specifically,
\begin{itemize}
\item if Spoiler makes a push move with $a_i$ (or $b_i$), Duplicator responds with a push move with $b_i$ (resp.~$a_i$). This results in a legal position, since neither $a_i$ nor $b_i$ is satisfied in any state that belongs to stage $i$;
\item if Spoiler makes a model move $p_{2i+1}\step p_{2i+2}$, Duplicator responds by $q_{2i+1}\step s_i$. The resulting position
\[
	\langle p_{2i+2},\vec{a},s_i,\vec{a}\cdot(a_i\ b_i)\rangle
\] 
is legal. Moreover, after the next round of model moves (possibly preceded by a sequence of push and pop moves) the game will necessarily reach the position
\[
	\langle p_{2i+3},\vec{c},p_{2i+3},\vec{c}\cdot(a_i\ b_i)\rangle
\] 
for some sequence $\vec{c}$. This position is winning for Duplicator, since neither $a_i$ nor $b_i$ are relevant for the remainder of the game.

Other model moves for Spoiler are dealt with in a similar way.
\end{itemize}

It remains to be shown that an opportunity to cheat must inevitably happen for Duplicator. Assume to the contrary that during a play, either $a_i$ or $b_i$ is remembered throughout the stage $i$, for each $i=1,\ldots,n-1$. Since model moves do not change the set of atoms remembered, it follows that $a_i$ or $b_i$ is remembered also in the last position before entering stage $i$, and in the first position after leaving stage $i$. As a result, positions in the play are covered by overlapping regions throughout which particular atoms must be remembered, as depicted here:
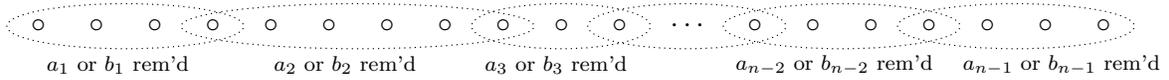
\[\xymatrix@=10pt{
\circ & \circ & \circ & \circ & \circ & \circ & \circ & \circ &\circ & \circ & \circ & \cdots &\circ & \circ & \circ & \circ &\circ & \circ & \circ 
\save "1,1"+<-12pt,0pt>."1,4"*+<.2cm>[o][F.]\frm{} \restore
\save "1,4"+<-12pt,0pt>."1,9"*+<.2cm>[o][F.]\frm{} \restore
\save "1,9"+<-12pt,0pt>."1,11"*+<.2cm>[o][F.]\frm{} \restore
\save "1,11"+<-12pt,0pt>."1,13"*+<.2cm>[o][F.]\frm{} \restore
\save "1,13"+<-12pt,0pt>."1,16"*+<.2cm>[o][F.]\frm{} \restore
\save "1,16"+<-12pt,0pt>."1,19"*+<.2cm>[o][F.]\frm{} \restore
\save "1,1"+<30pt,-15pt>*\txt{\scriptsize{$a_1$ or $b_1$ rem'd}}\restore
\save "1,4"+<50pt,-15pt>*\txt{\scriptsize{$a_2$ or $b_2$ rem'd}}\restore
\save "1,9"+<20pt,-15pt>*\txt{\scriptsize{$a_3$ or $b_3$ rem'd}}\restore
\save "1,13"+<30pt,-15pt>*\txt{\scriptsize{$a_{n-2}$ or $b_{n-2}$ rem'd}}\restore
\save "1,16"+<50pt,-15pt>*\txt{\scriptsize{$a_{n-1}$ or $b_{n-1}$ rem'd}}\restore
}\]
(The regions correspond to stages of the game, but note that the regions need not be of the same size, since the players may make different numbers of push-or-pop moves in each stage.) Since only $k$ atoms may be remembered at any given time, and since atoms are remembered in a stack-like regime, this is impossible if $n> 2^k$. To prove this, proceed by induction on $k$:
\begin{itemize}
\item for $k=1$ we have $n-1\geq 2$ so at some point two atoms must be remembered but there is no space for it,
\item for $k>1$, notice that throughout the overlapping regions, the stack that remembers atoms may never get empty, even for a moment. This means that there is some atom $a_i$ (or $b_i$) that sits at the bottom of the stack (formally, remains the first element of the atom sequence that is a component of game positions) throughout this part of the play. One of the intervals
\[
1,\ldots, i-1 \qquad \text{or} \qquad i+1,\ldots, n-1
\]
is of the size at least half of the interval $1,\ldots,n-1$. In the part of the play defined by the longer of the two intervals of stages, the atom $a_i$ (or $b_i$) remains useless, which effectively leaves enough space for remembering $k-1$ atoms. It is now enough to invoke the inductive assumption.
\end{itemize}

A winning strategy for Duplicator is to first copy Spoiler's moves, then use the first opportunity to cheat to force the game into a winning position as described above. This shows that Duplicator wins from the starting position $\langle p_1,\epsilon,q_1,\epsilon\rangle$ as required.

This means that $\pair{p_1}{\epsilon}$ and $\pair{q_1}{\epsilon}$ are $k$-stack-bisimilar, hence by Theorem~\ref{thm:k-stack-bisim} they satisfy the same globally $k$-supported formulas. Since every formula of the atomic $\mu$-calculus is globally $k$-supported for some $k$, this implies that no formula in $\eqaform$ defines \chain.

However, it is not hard to define \chain{} in the vectorial calculus $\veqaform$. Indeed, for any atom $a\in\Atoms$, the property {\em there exists a chain}
\[
	x_0\step x_1\step x_2 \step x_3 \step \cdots
\]
{\em starting in the current state such that $\pred(x_1)=\{a\}$}, is defined by a formula:
\[
	\phi_a = \nu X_a.\left\{X_b.\bigvee_{c\in\Atoms}\left(\theta_c\land\diam\left(\theta_b\land\diam X_c\right)\right)\right\}_{b\in\Atoms},
\]
where for any $b\in\Atoms$ the formula
\[
    \theta_b = b \land \bigwedge_{d\neq b}\neg d
\]
ensures that $b$, and no other basic proposition, holds in the current state.
Then \chain{} is defined by the alternative $\bigvee_{a\in\Atoms}\phi_a$.

\subsection{The \evensucc{} property}\label{sec:evensucc}{\ }

As before, consider a vocabulary of basic propositions that coincides with the set of $\Atoms$ of atoms. Let \evensucc{} denote the property ``the number of atoms $a$ such that some immediate successor of the current state satisfies $a$, is finite and even''. We shall show that this property is not definable in $\veqaform$ or $\ordaform$, but it is definable in $\vordaform$.

The argument for $\veqaform$ relies on the same model $\K$ as considered in Section~\ref{sec:infsucc}. Notice that $|S|$ there is even, hence the state $q\in K$ satisfies \evensucc{}, but the state $p\in K$ does not, and they are $k$-bisimilar over equality atoms as shown before.

For $\ordaform$, we rely on Theorem~\ref{thm:k-stack-bisim}. For any number $k$, we construct a model $\K$ over ordered atoms and two states $p,q\in K$ that are $k$-stack-bisimilar in $k$, but such that \evensucc{} holds for $p$ and fails for $q$. To this end, put $n=2^{k+1}$ and fix some finite sets $S,T\subseteq\Atoms$ with $|S|=n$, $|T|=2n+1$ and such that every element of $T$ is greater than every element of $S$. Let the set of states of $\K$ be:
\[
	K = \{p,q\} + \{r_a\mid a\in\Atoms\}
\]
with the obvious action of $\aut(\Atoms)$ that makes each of $p$ and $q$ a singleton orbit. The satisfaction relation of $\K$ is such that
\[
	\pred(p)=\pred(q)=\emptyset \qquad \text{and} \qquad \pred(r_a)=\{a\} \quad\text{for } a\in\Atoms,
\]
and the transition relation is defined by
\[
p\step r_a \text{ for }  a\in S \qquad \text{and} \qquad q\step r_a \text{ for } a\in T.
\]
Obviously the property \evensucc{} holds for $p$ but fails for $q$. However, $\pair{p}{\epsilon}$ and $\pair{q}{\epsilon}$ are $k$-stack-bisimilar according to Definition~\ref{def:k-bisim}. To see this we describe a winning strategy for Duplicator in the corresponding $k$-stack-bisimulation game.

Starting from the initial position $\langle p,\epsilon,q,\epsilon\rangle$, Spoiler may begin by making a series of push moves. An important observation is that at this stage, before any model move has been made, pop moves make no sense for Spoiler: they just bring the game to a position that has already been visited. This means that this first phase of the game takes at most $k$ moves. In this phase, Duplicator responds to Spoiler's moves much as she would in an $k$-round Ehrenfeucht-Fraiss\'e game on two large total orders, one of them of even size, the other of odd size. As in that classical scenario, Duplicator is able to survive $k$ rounds of the game, and ensure that any model move by the Spoiler afterwards can be met with a legal response.

As a result, by Theorem~\ref{thm:k-stack-bisim}, \evensucc{} is not definable in $\ordaform$. 

On the other hand, the property is definable in $\vordaform$. To see this, for fixed atoms $a<b$, consider the property ``$b$ is the least atom greater than $a$ such that some successor of the current state satisfies $b$'', defined by:
\[
	\phi^1_{a,b} = \diam b \land \bigwedge_{c\in(a;b)}\neg\diam c.
\]
Then define the property ``$b$ is the {\em second} least atom greater than $a$ such that some successor of the current state satisfies $b$'' by:
\[
	\phi^2_{a,b} = \bigvee_{c\in(a;b)}(\phi^1_{a,c}\land\phi^1_{c,b}).
\]
Then consider the property ``the number of atoms $c>a$ such that some immediate successor of the current state satisfies $c$, is finite and even'', defined by:
\[
	\psi_a = \mu X_a.\left\{X_b. \left(\bigwedge_{c>b}\neg\diam c\right)\lor\left(\bigvee_{c>b}(\phi^2_{b,c}\land X_c)\right)\right\}_{b\geq a}.
\]
Finally, \evensucc{} is defined by:
\[
	\bigvee_{a\in\Atoms}\bigwedge_{b<a}\psi_b.
\]


\section{Expresiveness limitations}\label{sect_limitations}
As before, consider Kripke models over a vocabulary of basic predicates that coincides with the set $\Atoms$ of atoms. Denote the property 
``there exists an infinite path where no $a$ holds more than once'', by \freshpath{}.
Such properties of states in Kripke models have potentially significant practical importance. For example, one may imagine a system equipped with a token (e.g. password) generator where one needs to verify that, on every path where no token is generated more than once, the security of the system is never breached. 

In a previous version of this paper~\cite{csl17}, we proved that \freshpath{} is not definable in the $\eqaform$, the scalar $\mu$-calculus over equality atoms. This is in contrast to the similar but definable property {\bf P2} from Section~\ref{sect_finite_model_property}. Our proof there did not work for ordered atoms, and in~\cite{csl17} we left is as an open problem whether the property is definable in the calculus that we here call $\ordaform$. We shall now show that, in fact, \freshpath{} is not definable even in the vectorial calculus $\vordaform$ over ordered atoms.

But first notice that there would be no hope of defining \freshpath{} in some calculus with decidable model checking if the property itself was undecidable. So we begin by showing that the property is indeed decidable on Kripke models over equality atoms, which makes its undefinability in our $\mu$-calculi all the more disappointing.

\begin{thm}\label{thm:freshpath-decid}
\freshpath{} is decidable on orbit-finite Kripke models over equality atoms.
\end{thm}
\begin{proof}
For simplicity, assume that a given orbit-finite Kripke model $\mathcal{K}$ is equivariant; a generalization to finitely supported models is straightforward.

Notice that for every state $x\in K$, the set $\pred(x)$ is either finite (and contained in $\supp(x)$) or co-finite. Moreover, a single orbit of $K$ only contains states of one of these two kinds. It is not difficult to decide the existence of a desired path where at least one state is of the second kind. Indeed, two such states cannot occur on the path at all, and even if exactly one of them occurs, almost all other states on the path must satisfy none of the basic predicates. The existence of such a path from a given state $x$ is straightforward to decide.

Once the existence of such paths is excluded, all (orbits of) states $x$ with co-finite $\pred(x)$ may be safely deleted from the model. From now on, assume that $\pred(x)\subseteq\supp(x)$ for each $x\in K$.

Derived from $\mathcal{K}$, construct a new orbit-finite Kripke model $\hat{\mathcal{K}}$, over the empty set of basic predicates, with the set of states defined by:
\[
	\hat{K} = \{\pair{x}{S} \mid x\in K,\  S\subseteq\supp(x)\setminus\pred(x)\} 
\]
and the transition relation by:
\[
	\pair{x}{S}\step^{\hat{\mathcal{K}}}\pair{y}{T} \text{ if and only if }  x\step^{\mathcal{K}}y \text{ and } (S\cup\pred(x))\cap\supp(y)\subseteq T
\]
(since there are no basic predicates, the satisfaction relation in $\hat{\mathcal{K}}$ is trivial).
The intuition is that in a state $\pair{x}{S}$, atoms in $S$ are marked as having had occurred previously on a path, and are forbidden from occurring in the future. Note that this marking is restricted to atoms from the support of the current state $x$ only.

Note that $\hat{\mathcal{K}}$ is indeed equivariant and orbit-finite: every orbit of $K$ gives rise to at most $2^k$ orbits in $\hat{K}$, where $k$ is the size of the least support of any (equivalently, every) element in the orbit.
Moreover, (a representation of) $\hat{\mathcal{K}}$ is computable from $\mathcal{K}$: for each orbit in $K$ one can enumerate all corresponding orbits in $\hat{K}$, and orbits of transitions are also easy to enumerate.

We shall now prove that a state $x\in K$ admits an infinite path where no $a$ holds more than once, if and only if $\pair{x}{\emptyset}\in\hat{K}$ admits any infinite path. (Note that the theorem follows from this, since the latter property is decidable, indeed, it is definable in the atomic $\mu$-calculus $\eqaform$, whose model-checking problem is decidable by Theorem~\ref{thm:mc-decid}).

For the left-to-right implication, assume an infinite path in $\mathcal{K}$, i.e., a sequence 
\[
x=x_0 \step^{\mathcal{K}} x_1 \step^{\mathcal{K}} x_2 \step^{\mathcal{K}} x_3 \step^{\mathcal{K}} \cdots
\]
such that $\pred(x_i)\cap\pred(x_j) = \emptyset$
for each $i\neq j\in\mathbb{N}$. Define
\[
	y_i = \pair{x_i}{S_i}, \quad \text{where} \quad
    S_i  = \supp(x_i)\cap\bigcup_{j=0}^{i-1}\pred(x_j).
\]
In particular, $y_0=\pair{x}{\emptyset}$. Then
\begin{equation}\label{eq:edge-Khat}
	\pair{x_i}{S_i}\step^{\hat{\mathcal{K}}}\pair{x_{i+1}}{S_{i+1}}
\end{equation}
for each $i\in\mathbb{N}$. Indeed, calculate
\begin{multline*}
(S_i\cup\pred(x_i))\cap\supp(x_{i+1})
= \left(\!\!\left(\supp(x_i)\cap\bigcup_{j=0}^{i-1}\pred(x_j)\right)\cup\pred(x_i)\right)\cap\supp(x_{i+1}) \\
\subseteq \left(\bigcup_{j=0}^{i}\pred(x_j)\right)\cap\supp(x_{i+1}) = S_{i+1}.
\end{multline*}
As a result
\begin{equation}\label{eq:seqn-Khat}
	\pair{x}{\emptyset}=\pair{x_0}{S_0}\step^{\hat{\mathcal{K}}}\pair{x_1}{S_1}\step^{\hat{\mathcal{K}}}\pair{x_2}{S_2}\step^{\hat{\mathcal{K}}}\cdots
\end{equation}
form an infinite path in $\hat{\mathcal{K}}$.

For the right-to-left implication, assume any infinite sequence as in~\eqref{eq:seqn-Khat}, for some $x_i$ and $S_i$ such that the condition~\eqref{eq:edge-Khat} holds for every $i\in\mathbb{N}$. We construct sequences
\[
y_0,y_1,\ldots \in K \quad\,
T_0,T_1,\ldots \subseteq\Atoms \quad\,
\pi_1,\pi_2,\ldots \in\aut(\Atoms)
\]
by simultaneous induction as follows: 
\begin{itemize}
\item $y_0=x_0$ and $T_0=S_0$,
\item $\pi_{i+1}$ is an atom automorphism such that: 
\begin{itemize}
\item $\pi_{i+1}(a)=a$ for $a\in\supp(y_i)$, and
\item $\pi_{i+1}(a)\not\in\bigcup_{j=0}^{i}\supp(y_j)$  
for $a\in\supp(x_{i+1}\cdot \pi_1\pi_2\cdots\pi_i)\setminus\supp(y_i)$,
\end{itemize}
and acting in an arbitrary way on the remaining atoms (such an automorphism always exists since the union of supports in the second clause is finite, and therefore some atoms exists outside of it),
\item $y_{i+1}=x_{i+1}\cdot \pi_1\pi_2\cdots\pi_i\pi_{i+1}$,
\item $T_{i+1}=S_{i+1}\cdot \pi_1\pi_2\cdots\pi_i\pi_{i+1}$.
\end{itemize}
Notice that, since $\pair{x_{i+1}}{S_{i+1}}$ is a legal state in $\hat{K}$, by equivariance so is $\pair{y_{i+1}}{T_{i+1}}$. Moreover,
\begin{equation*}
	\pair{y_i}{T_i}\step^{\hat{\mathcal{K}}}\pair{y_{i+1}}{T_{i+1}}
\end{equation*}
for each $i\in\mathbb{N}$, therefore the sequence
\begin{equation*}
	\pair{y_0}{T_0} \step^{\hat{\mathcal{K}}} \pair{y_1}{T_1} \step^{\hat{\mathcal{K}}} \pair{y_2}{T_2} \step^{\hat{\mathcal{K}}}\cdots
\end{equation*}
is an infinite path in $\hat{K}$. To see this, note that 
\[
	\pair{y_i\cdot\pi_{i+1}}{T_i\cdot\pi_{i+1}}\step^{\hat{\mathcal{K}}}\pair{y_{i+1}}{T_{i+1}},
\]
(by equivariance of the transition relation $\step^{\hat{\mathcal{K}}}$, applying the automorphism $\pi_1\cdots\pi_{i+1}$ to the transition 
$\pair{x_i}{S_i})\step^{\hat{\mathcal{K}}}\pair{x_{i+1}}{S_{i+1}}$), and
\[
	y_i\cdot\pi_{i+1}=y_i \quad\text{and}\quad T_i\cdot\pi_{i+1}=T_i
\]
since $\pi_{i+1}$ by definition fixes $\supp(y_i)$ and $T_i\subseteq\supp(y_i)$.

As a consequence, the sequence
\[
y_0\step^{\mathcal{K}} y_1 \step^{\mathcal{K}} y_2 \step^{\mathcal{K}} \cdots
\]
forms a path in $\mathcal{K}$. A useful property of this path, easy to infer from the definition of $y_i$, is that:
\begin{equation}\label{eq:loc-glob-fresh}	(\supp(y_{i+1})\setminus\supp(y_i))\cap\bigcup_{j=0}^{i}\supp(y_j)=\emptyset.
\end{equation}
In words, whenever a locally fresh atom appears in some $y_i$, then it does not appear anywhere earlier in the path.

We shall show that no basic predicate holds on this path more than once. Assume towards a contradiction that
\[	a\in \pred(y_i)\cap\pred(y_j)
\]
for some $a\in\Atoms$ and $i<j$. Then obviously $a\in\supp(y_i)$ and $a\in\supp(y_j)$, and by induction on the difference $j-i$, using~\eqref{eq:loc-glob-fresh}, we get that $a\in\supp(y_k)$ for all $k$ between $i$ and $j$. Again by induction, and by definition of the transition relation in $\hat{\mathcal{K}}$, $a$ belongs to all sets $T_{i+1},T_{i+2},\ldots,T_j$. But this means that $a\in T_j\cap\pred(y_j)$, which contradicts the fact that $\pair{y_j}{T_j}$ is a legal state in $\hat{K}$.
This completes the proof of the right-to-left implication, and of the entire theorem.
\end{proof}

We shall now show that, in spite of its decidability and intuitive simplicity, \freshpath{} is not definable in $\vordaform$, the vectorial atomic $\mu$-calculus over ordered atoms. To this end, we  rely on Theorem~\ref{thm:k-bisim}. We construct, for every number $k$, an orbit-finite Kripke model ${\mathcal{K}}$ over $\Atoms$ as the language of basic predicate symbols, and two states $x,y\in K$ such that $\pair{x}{\epsilon}$ and $\pair{y}{\epsilon}$ are $k$-bisimilar and \freshpath{} fails for $x$ but holds for $y$.

Given $k>0$, fix some $n>k$ and fix any set $S$ of $2n$ atoms:
\[
	a_1<b_1<a_2<b_2<\cdots<a_n<b_n.
\]
We shall build a Kripke model $\mathcal{K}$ and its sub-model $\mathcal{\check K}$, both supported by $S$.
Both models operate in three phases:
\begin{itemize}
\item Phase I, both in $\mathcal{K}$ and in $\mathcal{\check K}$, consists of $n$ states $p_1,\ldots,p_n\in K$ such that $\pred(p_i)=\{a_i\}$,
with transitions:
\[
	p_i\step p_{i+1} \qquad \text{for } 1\leq i<n.
\]
Note that the transition relation is entirely deterministic in this phase.
\item Phase II in $\mathcal{K}$ consists of $n^2+n$ states 
\[
	q^0_1,q^0_2,\ldots,q^0_n,q^1_1,q^1_2,\ldots,q^1_n,\ldots,q^n_1,q^n_2,\ldots,q^n_n \in K
\]
with
\[
	\pred(p^i_j) = \left\{
		\begin{array}{ll}
			\{a_j\} & \text{if } i=j \\
			\{b_j\} & \text{otherwise.}
		\end{array}\right.
\]
A transition from Phase I to Phase II is realised by adding $n+1$ transitions from $p_n$:
\[
	p_n \step q^i_1 \qquad \text{for } 0\leq i\leq n.
\]
Note that this introduces nondeterminism to the model. Further transitions in Phase II are:
\[
	q^i_j \step q^i_{j+1} \qquad \text{for } 0\leq i\leq n \text{ and } 1\leq j<n.
\]
These transitions, again, are deterministic.

In $\mathcal{\check K}$, Phase II is the same but with the states $q^0_1,q^0_2,\ldots,q^0_n$ excluded (so only $n^2$ states are present).

\item Phase III, both in $\mathcal{K}$ and in $\mathcal{\check K}$, is an infinite (but single-orbit) clique of states:
\[
	 \{r_c\mid c\in\Atoms\}\subseteq K
\]
with $\pred(r_c)=\{c\}$ and the full relation as the transition relation. Additionally, a transition from Phase II to Phase III is realised by adding a transition from each state $p^i_n$ (for $0\leq i\leq n$) to each state $r_c$. This, obviously, is highly nondeterministic.

Note that every state in $K$ satisfies exactly one basic predicate. The following diagram shows the general shape of $\mathcal{K}$, with basic predicates satisfied in a state marked over its outgoing transitions:
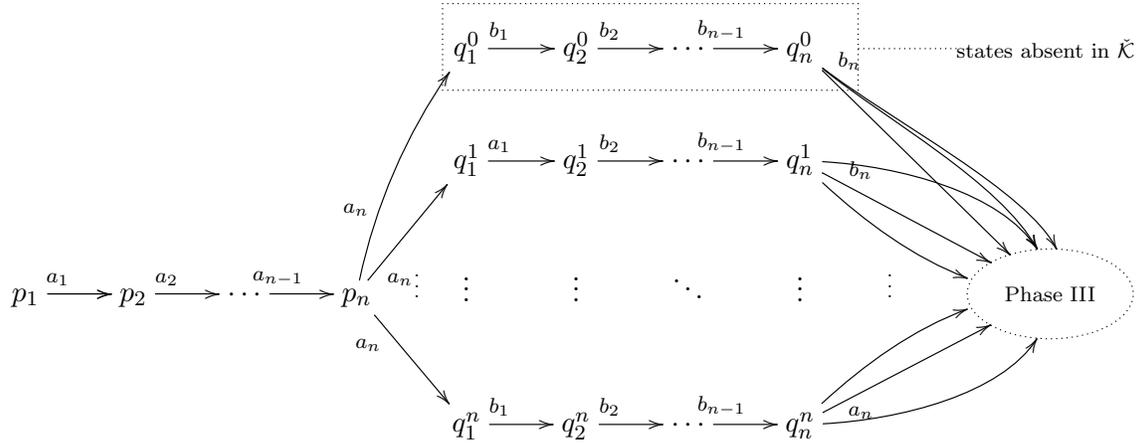
\[\xymatrix{
& & & & q^0_1\ar[r]^(.3){b_1} & q^0_2\ar[r]^(.3){b_2} &\cdots\ar[r]^(.3){b_{n-1}} & q^0_n\ar@(dr,u)[rrdd]\ar[rrdd]^(.12){b_n}\ar@(dr,ur)[rrdd] & \\
& & & & q^1_1\ar[r]^(.3){a_1} & q^1_2\ar[r]^(.3){b_2} & \cdots\ar[r]^(.3){b_{n-1}} & q^1_n\ar@(r,u)[rrd]\ar[rrd]^(.2){b_n}\ar@(dr,l)[rrd] \\
p_1\ar[r]\ar[r]^(.3){a_1} & p_2\ar[r]^(.3){a_2} & \cdots\ar[r]^(.3){a_{n-1}} & p_n\ar@{}@<-1ex>[r]^-{\tiny\vdots}\ar@/^/[ruu]^(.3){a_n}\ar[ru]_(.2){a_n}\ar[rd]_(.3){a_n} & \vdots & \vdots &\ddots &\vdots\ar@{}@<-1ex>[rr]^-{\tiny\vdots} & & 
*[o]+<1cm>[F.]{\scriptsize\text{Phase III}}\\
& & & & q^n_1\ar[r]^(.3){b_1} & q^n_2\ar[r]^(.3){b_2} & \cdots\ar[r]^(.3){b_{n-1}} & q^n_n\ar@(ru,l)[rru]\ar[rru]_(.2){a_n}\ar@(r,d)[rru]
\save "1,9"+<2cm,0cm>*\txt{\scriptsize{states absent in $\mathcal{\check K}$}}="tekst" \restore
\save "1,5"."1,8"*+<.5cm>[F.]\frm{}\ar@{.}"tekst" \restore
}\]
\end{itemize}

Both $\mathcal{K}$ and $\mathcal{\check K}$ are orbit-finite. Indeed, there are only finitely many states in Phases I and II, each of them forming a singleton orbit, and the single infinite orbit of states is formed by Phase III. The models are not equivariant, but with a small additional effort they can be made so: keeping $n$ fixed, consider a disjoint union of the models $\mathcal{K}$ taken for all sets of atoms $S=\{a_1,b_1,\cdots,a_n,b_n\}$. Then add a single equivariant state $\star$ where no basic predicate holds, with transitions from $\star$ to states $p_1$ in all the disjoint components. The resulting models remain orbit-finite, they are equivariant, and the following reasoning applies to them with little change.

It is easy to see that the state $p_1$ (and, indeed, every state) in $\mathcal{K}$ satisfies \freshpath: in the finite path
\[
	p_1,\ldots,p_n,q^0_1,\ldots,q^0_n
\]
no basic predicate holds more than once, and in Phase III it is easy to extend it to an infinite path so that this remains true. On the other hand, in the smaller model $\mathcal{\check K}$ the state $p_1$ does not satisfy \freshpath, since in Phase II one of the $a_i$ from Phase I is necessarily repeated there.

We shall show that $\pair{p_1}{\epsilon}$ in $\mathcal{K}$ is $k$-bisimilar to $\pair{p_1}{\epsilon}$ in $\mathcal{\check K}$ by providing a winning strategy for Duplicator in the corresponding $k$-bisimulation game. Formally, since we defined both $k$-bisimulations and $k$-bisimulation games as operating on single Kripke models, we need to consider the bisimulation game on a disjoint union of $\mathcal{K}$ and $\mathcal{\check K}$. To avoid confusion we will denote the states in the $\mathcal{\check K}$-component of this union by $\check{p}_i$, $\check{q}^i_j$ etc., and will describe a winning strategy for Duplicator from the position
\[
	\langle {p_1,\epsilon,\check{p}_1,\epsilon} \rangle
\]
in the $k$-bisimulation game.

In Phase I, Duplicator copies Spoiler's moves verbatim. That is, whenever Spoiler makes a model move from $p_i$ to $p_{i+1}$ (or from $\check{p}_i$ to $\check{p}_{i+1}$), Duplicator responds by moving from $\check{p}_i$ to $\check{p}_{i+1}$ (respectively, from $p_i$ to $p_{i+1}$; Duplicator has no choice here anyway, since in Phase I the transition relation is deterministic in both models). Moreover, whenever Spoiler makes an atom replacement move by choosing some new tuple of atoms, Duplicator responds by picking the same tuple of atoms. In this way, Duplicator ensures that in Phase I the game visits only positions of the form $\langle p_i,\vec{a},\check{p}_i,\vec{a} \rangle$ for $\vec{a}\in\Atoms^{(\leq k)}$.

Duplicator fares well this way until at some point Spoiler decides to make a model move from a configuration
\[
	\langle p_n,\vec{a},\check{p}_n,\vec{a} \rangle.
\]
Note that at this point Duplicator has no control over how atoms in $\vec{a}$ relate to the atoms $a_1,b_1,\cdots,a_n,b_n$.

If Spoiler makes a model move from $\check{p}_n$ to some $\check{q}^i_1$, Duplicator responds by moving from $p_n$ to $q^i_1$. The game is then essentially over: $\check{q}^i_1$ and $q^i_1$ have the same futures in $\mathcal{\check K}$ and $\mathcal{K}$, so Duplicator can proceed by forever copying Spoiler's moves. Similarly, if Spoiler makes a move from $p_n$ to some $q^i_1$ for $i>0$, Duplicator responds by moving from $\check{p}_n$ to $\check{q}^i_1$ and wins.

The only remaining case is when Spoiler moves from $p_n$ to $q^0_1$; the only move that does not have an obvious counterpart in $\mathcal{\check K}$. Since $n>k$ and $|\vec{a}|\leq k$, there exists an $i\in\{1,\ldots,n\}$ such that $\vec{a}$ contains no atom from the closed interval $[a_i,b_i]$. Then there exists an atom automorphism $\pi$ (i.e., a monotone bijection on the total order of rational numbers) such that:
\begin{itemize}
\item $\vec{a}\cdot\pi =\vec{a}$, that is, $\pi(c)=c$ for each $c\in\vec{a}$,
\item $\pi(b_j)=b_j$ for $j\neq i$, and
\item $\pi(b_i)=a_i$.
\end{itemize}
Having fixed such a $\pi$, Duplicator responds to Spoiler's move by moving from $\check{p}_n$ to $\check{q}^i_1$. Note that the target position:
\[
	\langle q^0_1,\vec{a},\check{q}^i_1,\vec{a} \rangle
\]
is legal. For this we need to check that:
\[
	\pair{\pred(q^0_1)}{\vec{a}}\sim\pair{\pred(\check{q}^i_1)}{\vec{a}}.
\]
By definition, $\pred(q^0_1)=\{b_1\}$. If $i\neq 1$ then $\pred(\check{q}^i_1)=\{b_1\}$ and the condition holds trivially. If $i=1$ then $\pred(\check{q}^1_1)=\{a_1\}$ and the condition holds since $\vec{a}$ does not contain any atom from the interval $[a_1,b_1]$.

Afterwards, Duplicator copies Spoiler's moves translating them along the automorphism $\pi$. More specifically, throughout Phase II:
\begin{itemize}
\item if Spoiler makes a model move from $q^0_j$ to $q^0_{j+1}$, Duplicator responds by moving from $\check{q}^i_j$ to $\check{q}^i_{j+1}$ (again, no real choice here) and {\em vice versa},
\item if Spoiler makes an atom replacement move by choosing a new tuple of atoms $\vec{c}$ in the $\mathcal{K}$-component of the game position, Duplicator responds by choosing $\vec{c}\cdot\pi$. If, on the other hand, Spoiler picks a new tuple of atoms $\vec{d}$ in the $\mathcal{\check K}$-component, Duplicator responds by choosing $\vec{d}\cdot\pi^{-1}$.
\end{itemize}
In Phase III the strategy remains similar, except that a Spoiler's model move to a state $r_c$ is matched by Duplicator's move to $\check{r}_{\pi(c)}$, and a Spoiler's move to $\check{r}_c$ is matched by a move to $r_{\pi^{-1}(c)}$. This strategy ensures that only legal positions in the game are visited, and since every move by Spoiler is met with a response, the strategy is winning for Duplicator.

This means that $\pair{p_0}{\epsilon}$ and $\pair{\check{p}_0}{\epsilon}$ are $k$-bisimilar, hence by Theorem~\ref{thm:k-bisim} they satisfy the same globally $k$-supported formulas. Since every formula of the atomic $\mu$-calculus is globally $k$-supported for some $k$, this implies that no formula defines \freshpath.


\section{Future work}\label{sect_conclusions}
We list some interesting aspects of the atomic $\mu$-calculus that we leave for future work.

\smallskip
\noindent
{\bf Complexity issues.} For the decidability results we presented, in particular for the model checking problem over orbit-finite structures, one immediately asks about the complexity of the algorithms proposed. The answer depends on the way one measures the size of input structures. One obvious option is to consider the length of their representation with set-builder expressions and first-order formulas. With this view, most basic operations listed in Remark~\ref{rem:computability} become {\sc Pspace}-hard, because the first-order theory of pure equality is {\sc Pspace}-complete. As a result, the complexity of basic operations dwarfs the distinction between the two algorithmic approaches to model checking based on direct fixpoint computation and on parity games. 

Another approach is to measure orbit-finite structures by the number of their orbits, and the (hereditary) size of their least support. Note that the number of orbits of a set can be exponentially bigger than the size of its logical representation; for example, the number of orbits of the set $\Atoms^n$, whose representation has size linear in $n$, is equal to the $n$-th Bell number. In this view the difference between the two approaches to model checking becomes more prominent.

We defer precise complexity analyses until we have a better general understanding of various time and space complexity models on atomic structures.

\smallskip
\noindent
{\bf Other atoms.} 
To simplify the presentation, in this paper we focus on equality atoms and ordered atoms only. However, the definition of atomic $\mu$-calculus and its basic properties could be transported without much difficulty to other relational structures of atoms subject to some model-theoretic assumptions such as oligomorphicity (i.e. the assumption that $\Atoms^n$ is orbit-finite for every $n$), homogeneity (i.e. the assumption that every isomorphism between two finite substructures of  $\Atoms$ extends to an automorphism of $\Atoms$), and decidability of the first-order theory of $\Atoms$.

This, however, is with some exceptions. Some of our proofs in this paper, in particular the one for undefinability of \freshpath{}, rely on particular atom structures and it is not clear how far they generalize. As a concrete open problem, we leave the question whether \freshpath{} is perhaps definable over atoms from
the universal undirected graph, also known as the {\em random graph} or {\em Rado graph}, where vertices are natural numbers, and an undirected edge $\{n,m\}$ is present if and only if the $n$-th bit in the binary representation of $m$ is $1$ (for $n < m$). 

\smallskip
\noindent
{\bf Defining \freshpath{}.} The fact that \freshpath{} is not definable in atomic $\mu$-calculi is disappointing, since it looks like a property of potential practical importance in system verification. 
Since we know that \freshpath{} is decidable, it is desirable to extend atomic $\mu$-calculus in some well-structured and syntactically economic way that would allow one to define such properties  while preserving the decidability of model checking. The property of ``global freshness'' has been studied in the context of automata with atoms~\cite{tzevelekos}, and one may look for inspiration there. Our proof of Theorem~\ref{thm:freshpath-decid} also suggests some promising options. We leave this for future work.

\bibliographystyle{plain}
\bibliography{main}

\end{document}